\RequirePackage{fix-cm}

\documentclass[smallextended]{svjour3}

\smartqed

\usepackage{graphicx}

\usepackage{amsmath}
\usepackage{amssymb}
\usepackage{braket}
\usepackage{color}

\newtheorem{thm}{Theorem}

\newtheorem{lem}{Lemma}

\begin{document}

\title{A quantum walk on the half line with a particular initial state}
\subtitle{}

\titlerunning{A quantum walk on the half line with a particular initial state}

\author{Takuya Machida}

\authorrunning{T.~Machida}

\institute{%
T.~Machida \at
              College of Industrial Technology, Nihon University, Narashino, Chiba 275-8576, Japan\\
              \email{machida.takuya@nihon-u.ac.jp}
}

\date{}

\maketitle

\begin{abstract}
Quantum walks are considered to be quantum counterparts of random walks.
They show us impressive probability distributions which are different from those of random walks.
That fact has been precisely proved in terms of mathematics and some of the results were reported as limit theorems.
When we analyze quantum walks, some conventional methods are used for the computations.
Especially, the Fourier analysis has played a role to do that.
It is, however, compatible with some types of quantum walks (e.g. quantum walks on the line with a spatially homogeneous dynamics) and can not well work on the derivation of limit theorems for all the quantum walks.
In this paper we try to obtain a limit theorem for a quantum walk on the half line by the usage of the Fourier analysis.
Substituting a quantum walk on the line for it, we will lead to a possibility that the Fourier analysis is useful to compute a limit distribution of the quantum walk on the half line.
\end{abstract}

\section{Introduction}
As a unitary process, quantum walks are expected to describe quantum systems and to be useful to aspect the behavior of such systems.
Quantum walks are also called quantum random walks, but their behavior is largely different from that of random walks.
While probability distributions of random walks are generally diffusive and expressed by Gaussian distributions in approximation, those of quantum walks show ballistic actions and their asymptotic behavior is far from Gaussian distributions.
That fact has been proved in mathematics and has resulted in limit theorems~\cite{Venegas-Andraca2012}.
One of the ways to derive the limit theorems is the Fourier analysis.
It, however, is not that the method has advantage for all the quantum walks.
The Fourier analysis well works on analysis of the quantum walks whose dynamics are spatially homogeneous.
For instance, some researchers succeeded in deriving long-time limit distributions for quantum walks on the line or on the plane.
On the other hand, for spatially inhomogeneous quantum walks, limit theorems are hard to get by Fourier analysis due to the inhomogeneity of the corresponding time evolution protocol in Fourier picture.
In this paper we study a quantum walk on the half line whose dynamics is not spatially homogeneous, and challenge to derive a limit theorem (Theorem \ref{th:151212} in page \pageref{th:151212}) for it by Fourier analysis.
Under a particular initial state, we can make a copy of the quantum walk on the half line to a quantum walk on the line.
The study of an alternate system of a quantum walk was also reported for a walk on the plane \cite{DiMcBusch2011}.
In the study, it was proved that there existed a two-state quantum walk on the plane whose probability distribution was same as that of a four-state quantum walk on the plane.
Its result was expanded later and a limit distribution of the alternate two-state quantum walk was demonstrated \cite{DiMcMachidaBusch2011}.

Back to the quantum walk on the half line, using the substitutional system, we will get an exact representation for the probability distribution of the walk and be accessible to its limit distribution by Fourier analysis.
A representation for the probability distribution was exactly estimated just for a quantum walk on the line \cite{Konno2002a} in which its limit distribution was also proved, and the exact representation for the quantum walk on the half line has never been seen before.
Although the limit distribution for the quantum walk on the half line was already obtained in a manner different from the Fourier analysis \cite{LiuPetulante2013}, we will demonstrate the derivation by Fourier analysis.
The paper \cite{LiuPetulante2013} took care of a quantum walk which contained the quantum walk defined in this paper.
From the result in the past study, we can tell that there is a class of quantum walk on the half line whose probability distribution has the possibility of localization.
But, our quantum walk is out of the class of localization, that is, we will not see localization at all in the rest of this paper.
We also observe localization for quantum walks on the half line in Refs. \cite{KonnoSegawa2011,KonnoSegawa2014} in which the analysis was based on the CGMV method \cite{CanteroMoralGrunbaumVelazquez2010}.
The CGMV method is known as a way to mention the possibility of localization of quantum walks.  

In the subsequent section, we define a quantum walk on the half line which we should study in this paper.
Then, a substitutional quantum walk on the line is introduced and we see two results for it in Sec. \ref{sec:L_QW}.
One is an exact representation for the probability distribution of the walk, and the other is a limit distribution.
Since both can be derived by the results of past studies, they are briefly shown without any proof.
Using the results for the quantum walk on the line, we can give an exact representation for the probability distribution of the quantum walk on the half line and a limit distribution for it.
In the final section we conclude this paper and make a discussion.

\section{Quantum walk on the half line with a localized initial state}
\label{sec:HL_QW}

We first define a quantum walk on the half line.
With two inner states, the quantum walker moves on the half line.
Its total system is described as the superposition of the position states and the inner states, and the system spreads out across space as a wave.
We present the system of the walk in a tensor Hilbert space.
The position states are given in a Hilbert space $\mathcal{H}_p^{HL}$ spanned by the orthonormal basis $\left\{\ket{x} : x\in\left\{0,1,2,\ldots\right\}\right\}$.
On the other hand, the inner states are expressed in a Hilbert space $\mathcal{H}_c$ spanned by the orthonormal basis $\left\{\ket{0}, \ket{1}\right\}$.
Let $\ket{\Psi_t}\in\mathcal{H}_p^{HL}\otimes\mathcal{H}_c$ be the system of quantum walk on the half line at time $t\,\in\left\{0,1,2,\ldots\right\}$.
It is updating in a unitary process featured by a unitary operator which is casted to the elements in the Hilbert space $\mathcal{H}_c$,
\begin{align}
 C=&\cos\theta\ket{0}\bra{0}+\sin\theta\ket{0}\bra{1}+\sin\theta\ket{1}\bra{0}-\cos\theta\ket{1}\bra{1}\quad (\theta\in [0,2\pi))\nonumber\\
 =&c\ket{0}\bra{0}+s\ket{0}\bra{1}+s\ket{1}\bra{0}-c\ket{1}\bra{1}.\label{eq:coin-flip_operator}
\end{align}
In Eq. \eqref{eq:coin-flip_operator}, the letter $c$ (resp. $s$) is short for $\cos\theta$ (resp. $\sin\theta$).
Supposing that the quantum walk is in a localized superposition state at the initial time 0,
\begin{equation}
 \ket{\Psi_0}=\ket{0}\otimes e^{-i\theta}\left(\frac{1}{\sqrt{2}}\ket{0}+\frac{i}{\sqrt{2}}\ket{1}\right),\label{eq:HL_initial_state}
\end{equation}
it updates to the next state with a unitary evolution
\begin{equation}
 \ket{\Psi_{t+1}}=\tilde{S}^{HL}\tilde{C}^{HL}\ket{\Psi_t},\label{eq:HL_time-evolution}
\end{equation}
where
\begin{align}
 \tilde{C}^{HL}=&\sum_{x=0}^\infty\ket{x}\bra{x}\otimes C,\\
 \tilde{S}^{HL}=&\ket{0}\bra{0}\otimes\ket{1}\bra{0}+\sum_{x=1}^\infty\ket{x-1}\bra{x}\otimes\ket{0}\bra{0}+\sum_{x=0}^\infty\ket{x+1}\bra{x}\otimes\ket{1}\bra{1}.\label{eq:HL_shift_operator}
\end{align}
Due to the presence of the operator $\tilde{S}^{HL}$, the dynamics of the quantum walk is not spatially homogeneous,
\begin{align}
 \tilde{S}^{HL}(\ket{x}\otimes\ket{0})=&\left\{\begin{array}{ll}
				       \ket{x}\otimes\ket{1} & (x=0),\\
					      \ket{x-1}\otimes\ket{0} & (x=1,2,\ldots),
					     \end{array}\right.\\
 \tilde{S}^{HL}(\ket{x}\otimes\ket{1})=&\ket{x+1}\otimes\ket{1}\quad (x=0,1,2,\ldots).
\end{align}

The walker in the inner state $j\in\left\{0,1\right\}$ is observed at position $x\in\left\{0,1,2,\ldots\right\}$ at time $t$ with the probability
\begin{equation}
 \mathbb{P}(X_t^{HL}=x;j)=\bra{\Psi_t}\Bigl(\ket{x}\bra{x}\otimes\ket{j}\bra{j}\Bigr)\ket{\Psi_t},\label{eq:probability_inner_state}
\end{equation}
where $X_t^{HL}$ denotes the position of the quantum walker on the half line at time $t$.
In quantum walks, the probability distribution regardless of the inner states is normally analyzed rather than the probability in Eq. \eqref{eq:probability_inner_state},
\begin{equation}
 \mathbb{P}(X_t^{HL}=x)=\sum_{j=0}^1 \mathbb{P}(X_t^{HL}=x;j)=\bra{\Psi_t}\biggl(\ket{x}\bra{x}\otimes\sum_{j=0}^1\ket{j}\bra{j}\biggr)\ket{\Psi_t}.\label{eq:HL_probability}
\end{equation}
Note that the probability distribution outputs the same values of probability even when we set another localized initial state $\ket{\Psi_0}=\ket{0}\otimes \left(1/\sqrt{2}\ket{0}+i/\sqrt{2}\ket{1}\right)$, which has been often employed in the past studies on quantum walks.
That fact comes from the linearity of the updating rule in Eq. \eqref{eq:HL_time-evolution} and the definition of the probability distributions in Eq. \eqref{eq:probability_inner_state}.
As shown in Figs. \ref{fig:160201_01} and  \ref{fig:160201_04}, the probability distributions have a sharp peak and a ballistic behavior.
From Fig. \ref{fig:160201_07}, it is figured out how the sharp peak depends on the parameter $\theta$ which is embedded in the unitary operator $C$ in Eq. \eqref{eq:coin-flip_operator}.

\begin{figure}[h]
\begin{center}
 \begin{minipage}{35mm}
  \begin{center}
   \includegraphics[scale=0.3]{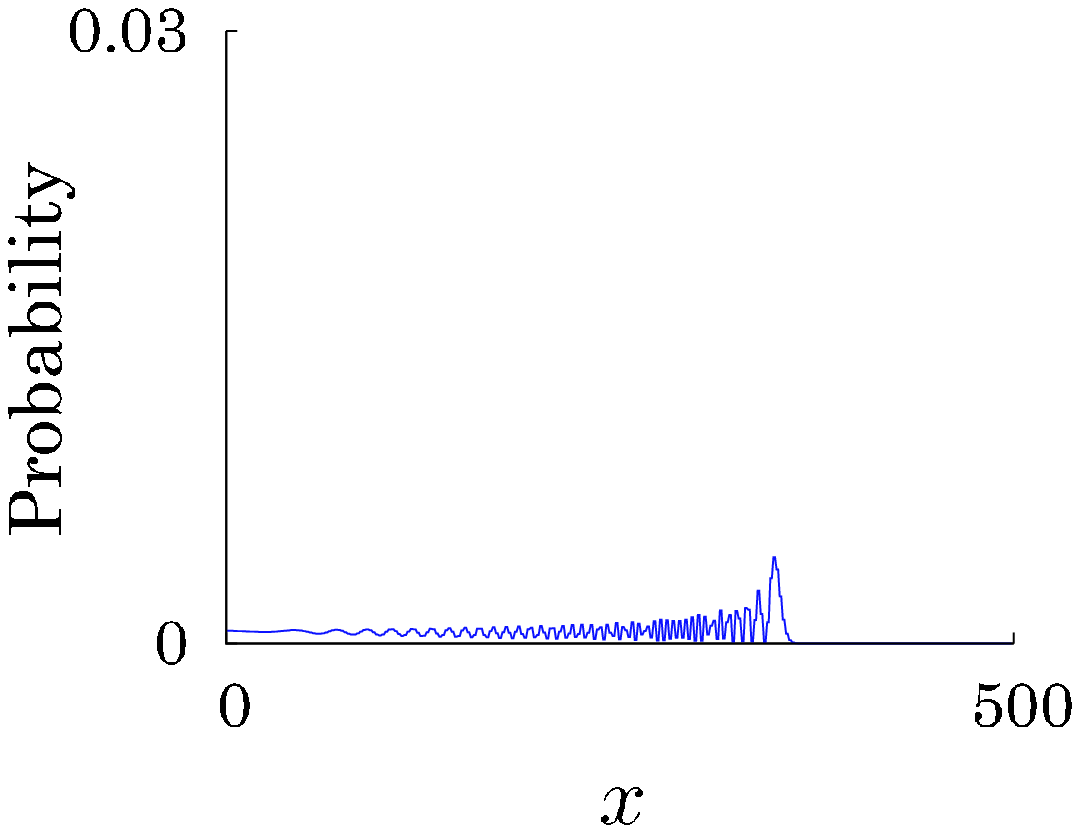}\\[2mm]
  (a) $\mathbb{P}(X_{500}^{HL}=x;0)$
  \end{center}
 \end{minipage}
 \begin{minipage}{35mm}
  \begin{center}
   \includegraphics[scale=0.3]{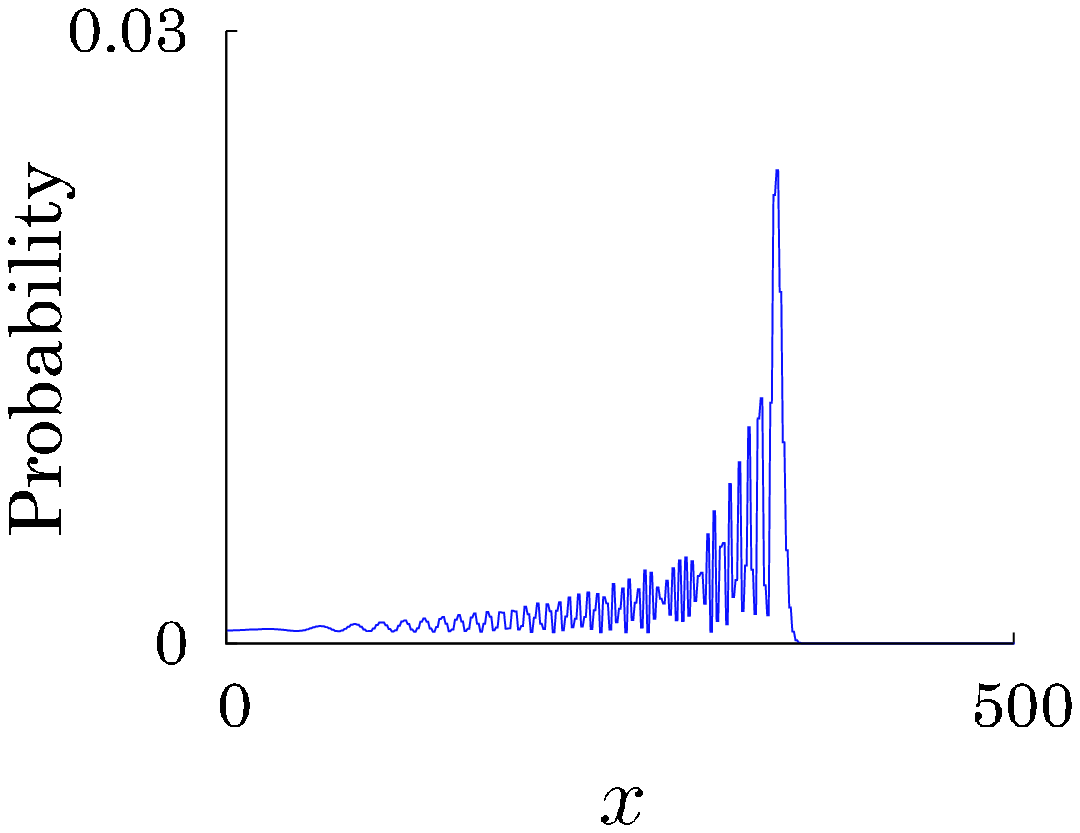}\\[2mm]
  (b) $\mathbb{P}(X_{500}^{HL}=x;1)$
  \end{center}
 \end{minipage}
 \begin{minipage}{35mm}
  \begin{center}
   \includegraphics[scale=0.3]{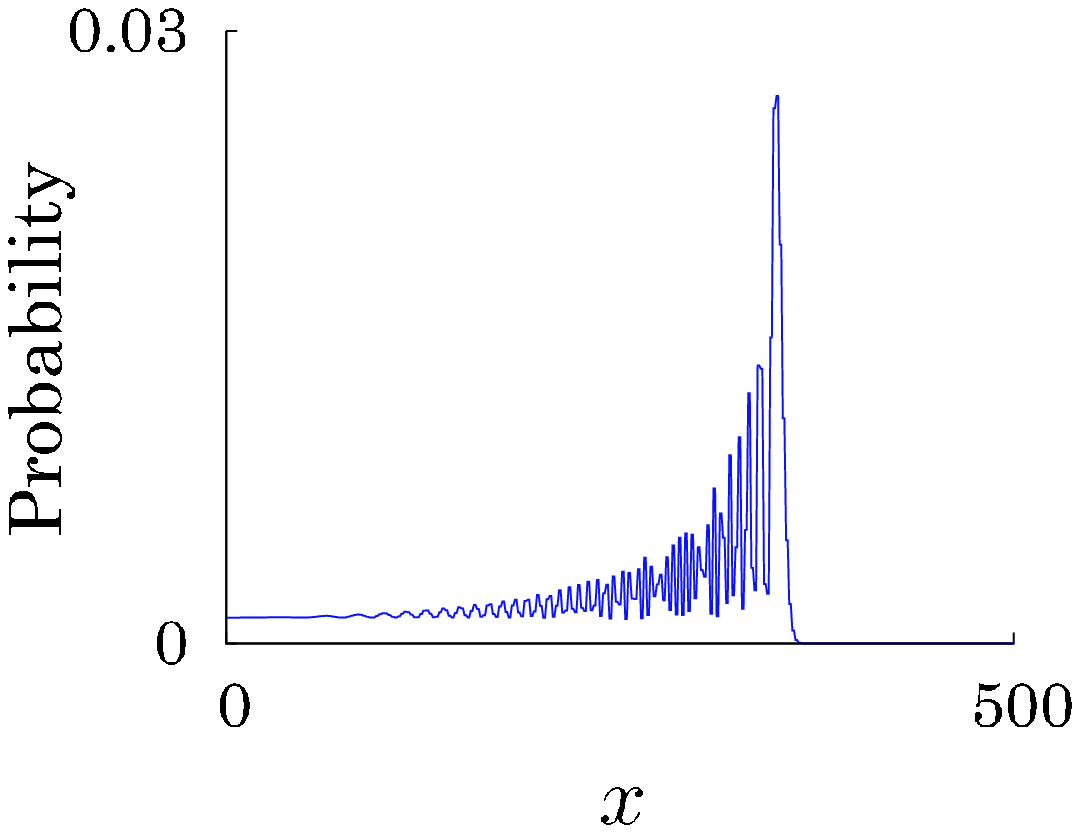}\\[2mm]
  (c) $\mathbb{P}(X_{500}^{HL}=x)$
  \end{center}
 \end{minipage}
\vspace{5mm}
\caption{$\theta=\pi/4$ : We observe a sharp peak in each distribution. Immediately after passing the peak toward the opposite way to the origin, the probability plummets.}
\label{fig:160201_01}
\end{center}
\end{figure}

\begin{figure}[h]
\begin{center}
 \begin{minipage}{35mm}
  \begin{center}
   \includegraphics[scale=0.15]{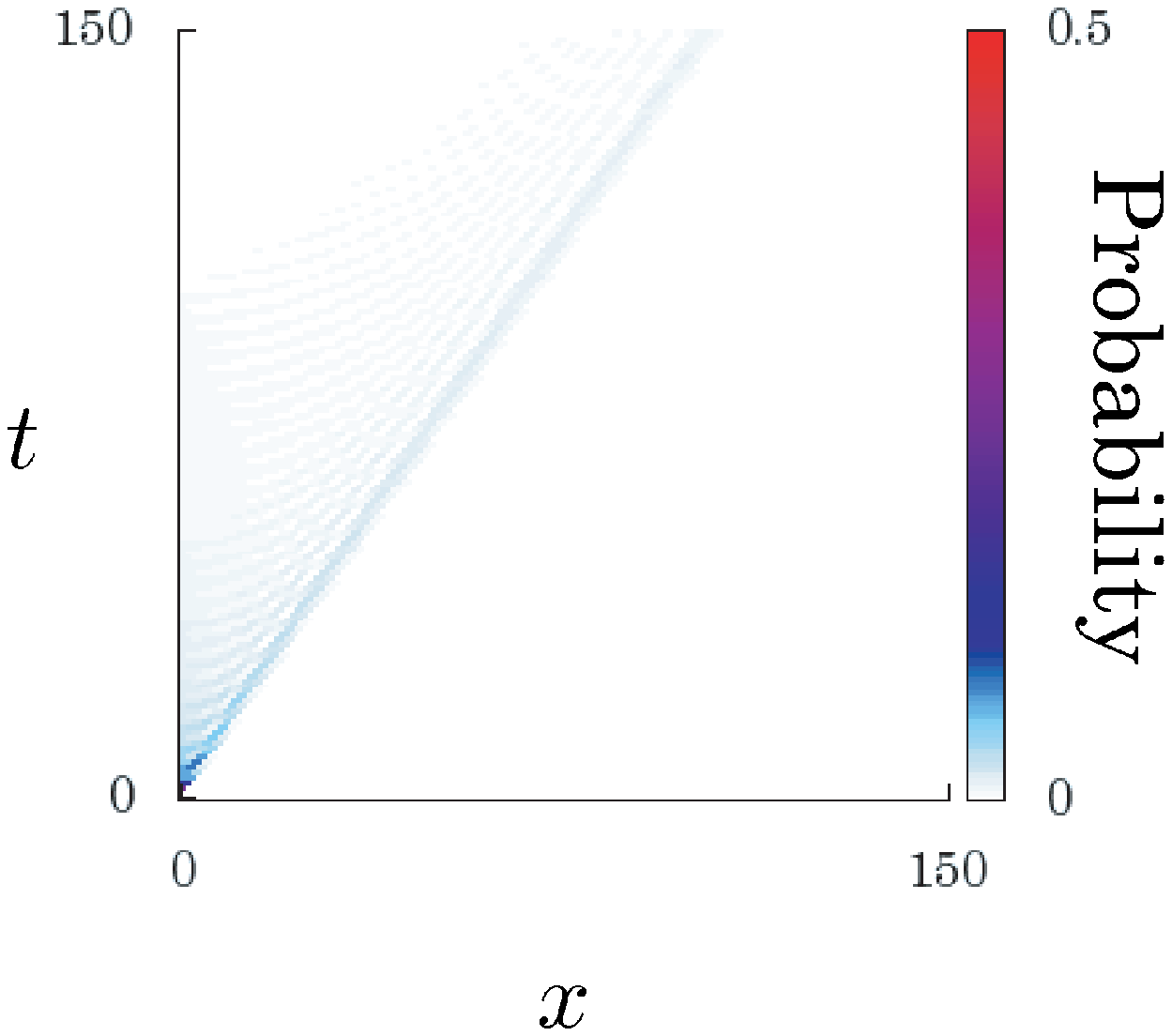}\\[2mm]
  (a) $\mathbb{P}(X_t^{HL}=x;0)$
  \end{center}
 \end{minipage}
 \begin{minipage}{35mm}
  \begin{center}
   \includegraphics[scale=0.15]{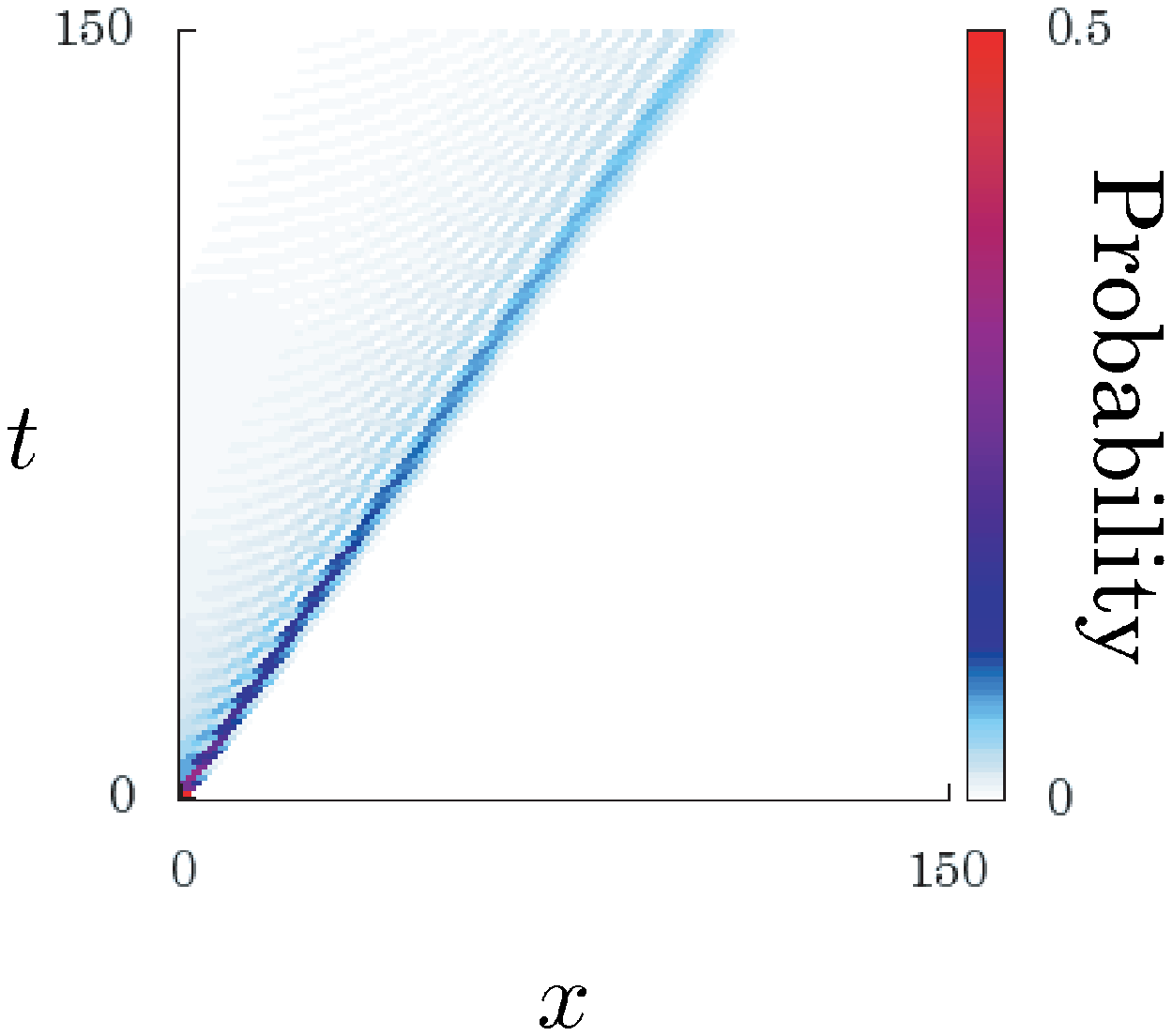}\\[2mm]
  (b) $\mathbb{P}(X_t^{HL}=x;1)$
  \end{center}
 \end{minipage}
 \begin{minipage}{35mm}
  \begin{center}
   \includegraphics[scale=0.15]{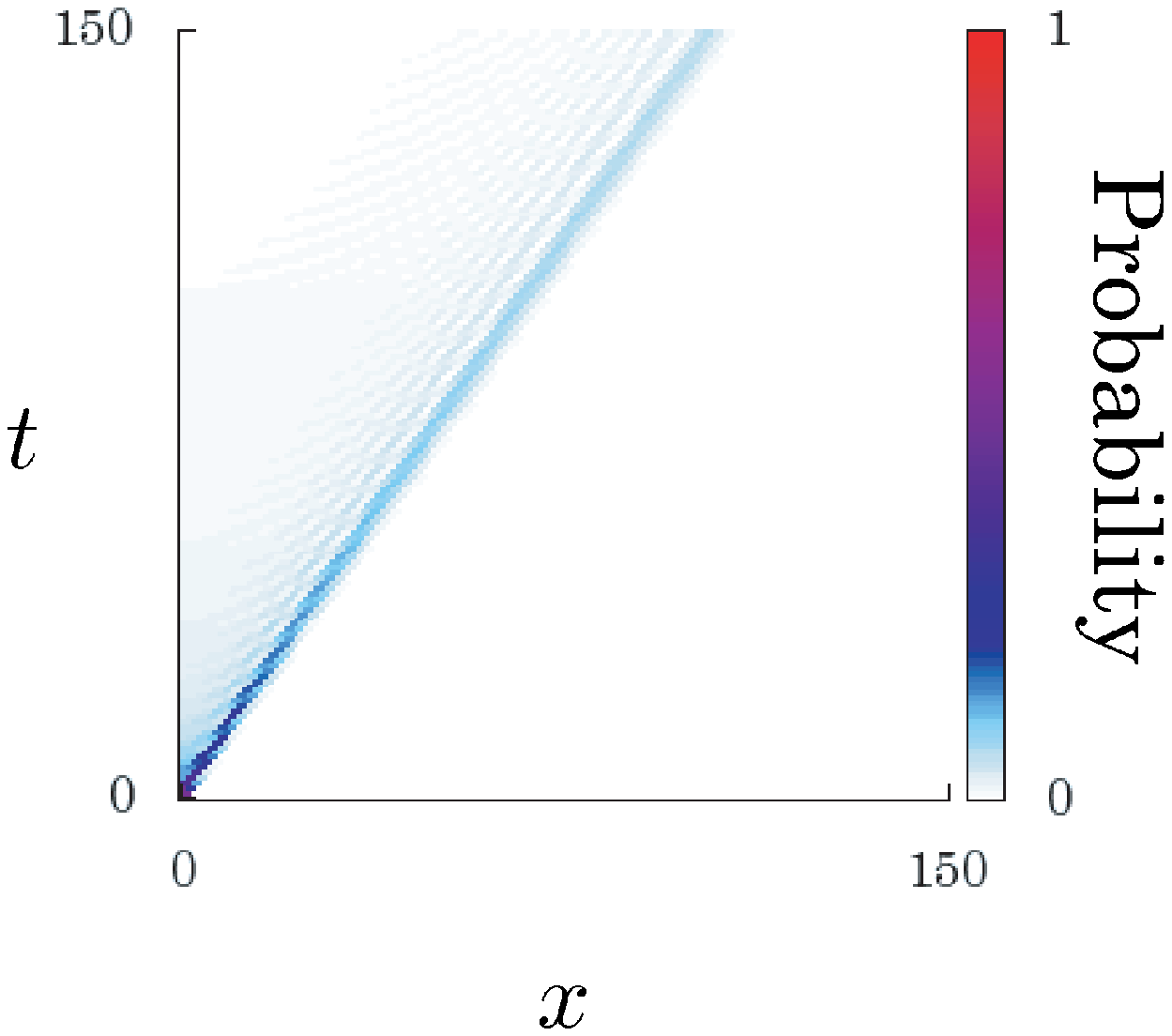}\\[2mm]
  (c) $\mathbb{P}(X_t^{HL}=x)$
  \end{center}
 \end{minipage}
\vspace{5mm}
\caption{$\theta=\pi/4$ : Each distribution is spreading out as time $t$ goes up, and its behavior is ballistic. While the maximum value of the probability is 0.5 in Figs. (a) and (b), it is 1 in Fig. (c).}
\label{fig:160201_04}
\end{center}
\end{figure}

\begin{figure}[h]
\begin{center}
 \begin{minipage}{35mm}
  \begin{center}
   \includegraphics[scale=0.15]{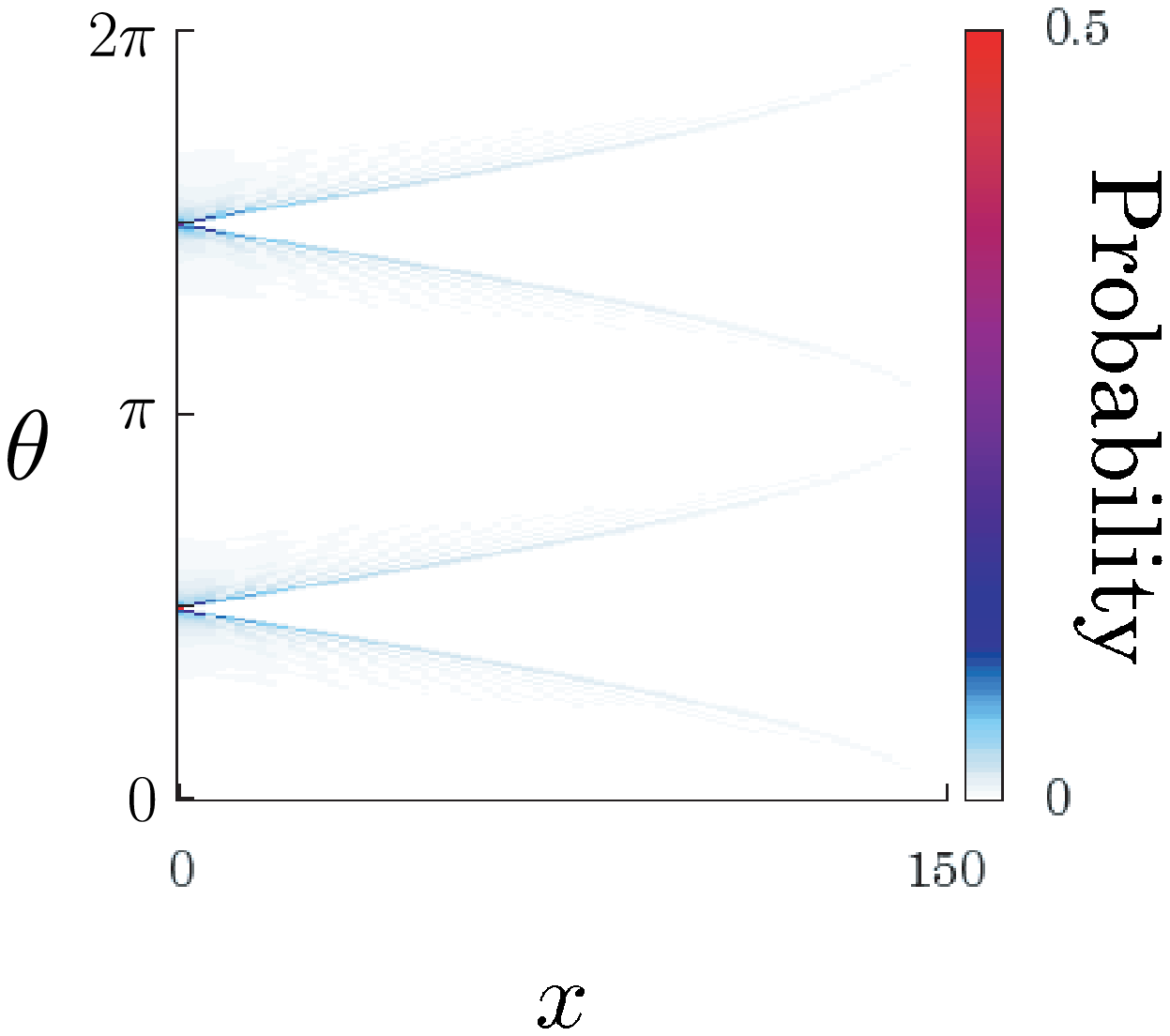}\\[2mm]
  (a) $\mathbb{P}(X_{150}^{HL}=x;0)$
  \end{center}
 \end{minipage}
 \begin{minipage}{35mm}
  \begin{center}
   \includegraphics[scale=0.15]{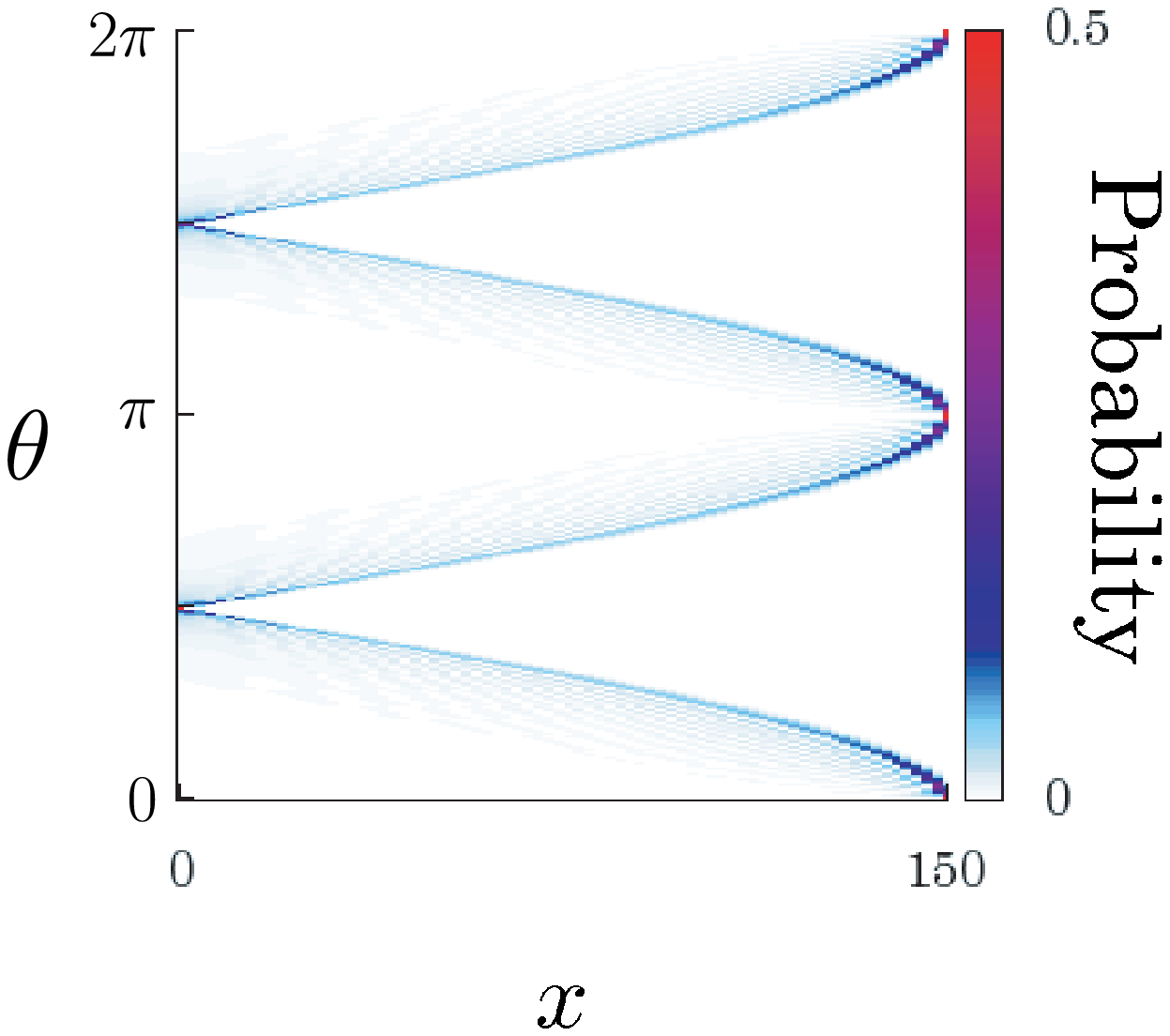}\\[2mm]
  (b) $\mathbb{P}(X_{150}^{HL}=x;1)$
  \end{center}
 \end{minipage}
 \begin{minipage}{35mm}
  \begin{center}
   \includegraphics[scale=0.15]{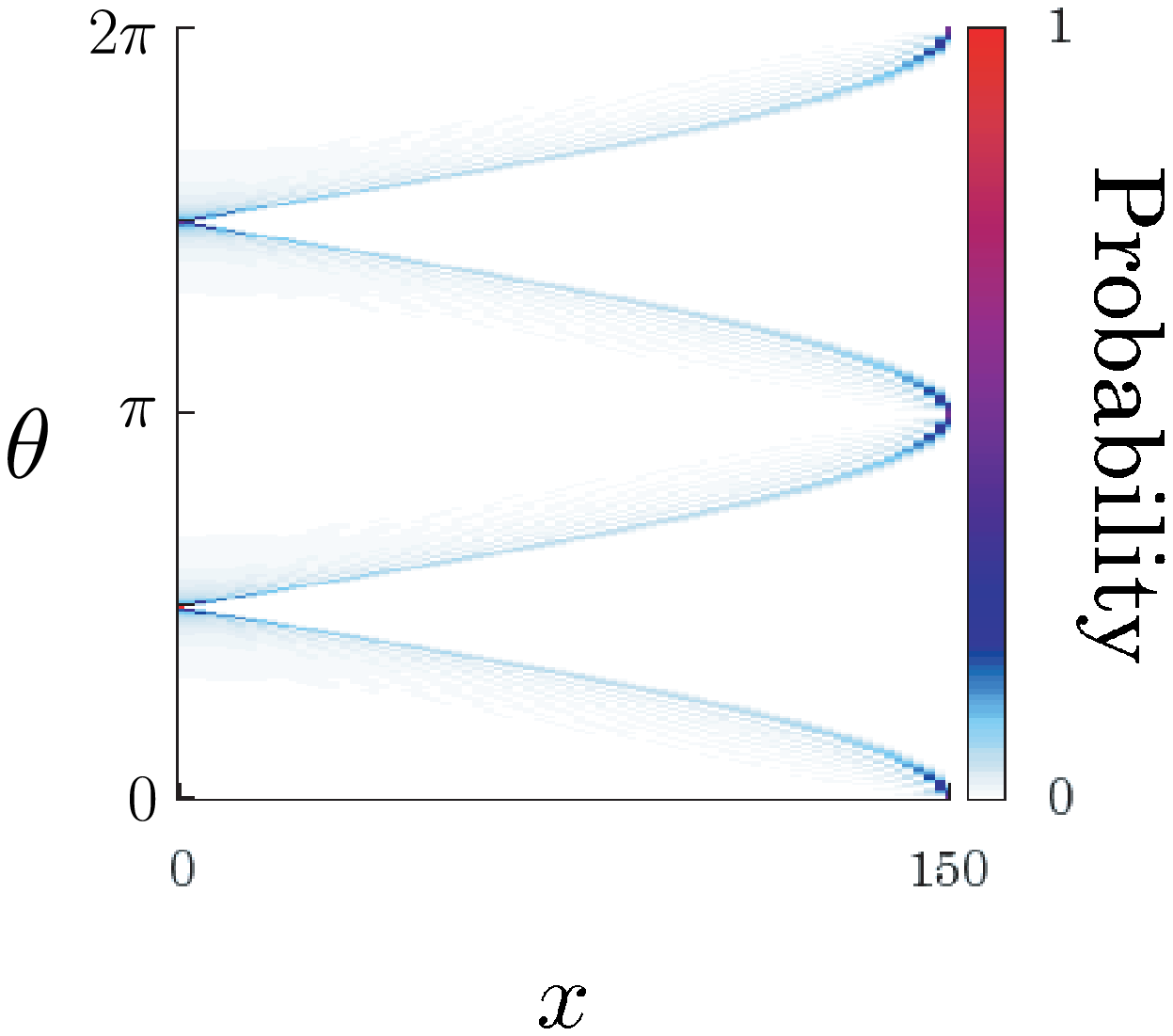}\\[2mm]
  (c) $\mathbb{P}(X_{150}^{HL}=x)$
  \end{center}
 \end{minipage}
\vspace{5mm}
\caption{These pictures tell us how the distributions depend on the parameter $\theta$. While the maximum value of the probability is 0.5 in Figs. (a) and (b), it is 1 in Fig. (c).}
\label{fig:160201_07}
\end{center}
\end{figure}

\section{Substitutional quantum walk on the line with a delocalized initial state}
\label{sec:L_QW}

In this section we introduce a standard quantum walk on the line.
Setting a delocalized initial state, we will tell that the quantum walk on the line makes a complete copy of the quantum walk on the half line which was defined in Sec. \ref{sec:HL_QW}.
The system of standard quantum walk is also described in a tensor Hilbert space $\mathcal{H}_p^{L}\otimes\mathcal{H}_c$, where the Hilbert space $\mathcal{H}_p^L$ is spanned by the orthonormal basis $\left\{\ket{x} : x\in\mathbb{Z}\right\}$.
The Hilbert space $\mathcal{H}_p^L$ represents the space where the quantum walker moves, which means that the walker is spreading out on the line $\mathbb{Z}=\left\{0,\pm 1,\pm 2,\ldots\right\}$ as time $t$ is going up.
The update rule of the walk, represented by $\ket{\Phi_t}\in\mathcal{H}_p^L\otimes\mathcal{H}_c$ at time $t$, is almost same as the one assigned to the quantum walk on the half line, 
\begin{equation}
 \ket{\Phi_{t+1}}=\tilde{S}^L\tilde{C}^L\ket{\Phi_t},\label{eq:L_time-evolution}
\end{equation}
where
\begin{align}
 \tilde{C}^L=&\sum_{x\in\mathbb{Z}}\ket{x}\bra{x}\otimes C,\\
 \tilde{S}^L=&\sum_{x\in\mathbb{Z}}\ket{x-1}\bra{x}\otimes\ket{0}\bra{0}+\ket{x+1}\bra{x}\otimes\ket{1}\bra{1}.
\end{align}
Looking at the form of the operator $\tilde{S}^L$, we realize that the dynamics of this quantum walk is spatially homogeneous.

Let us suppose that the walker begins to update with a delocalized initial state
\begin{equation}
 \ket{\Phi_0}=\Bigl(\ket{-1}+\ket{0}\Bigr)\otimes\left(\frac{\cos\theta}{\sqrt{2}}\ket{0}+\frac{\sin\theta}{\sqrt{2}}\ket{1}\right).\label{eq:L_initial_state}
\end{equation}
The parameter $\theta$ in Eq. \eqref{eq:L_initial_state} takes the same value of the parameter $\theta$ which gives the unitary operator $C$ in Eq. \eqref{eq:coin-flip_operator}.
The large reason we have picked this initial state is that we substitute the quantum walk on the line for the quantum walk on the half line.
In the view of mathematics, the behavior of quantum walks as $t\to\infty$ is interesting and a lot of long-time limit theorems have been studied because they are capable of accounting for the asymptotic behavior of the walkers as time $t$ becomes large enough.
When we assert limit theorems for quantum walks, their proofs are always required.
As one of the methods to demonstrate them, the Fourier analysis has played an important role.
Although the Fourier analysis is compatible with limit theorems for quantum walks on the line, it is not so accessible to limit theorems for quantum walks on the half line.
But, thanks to the substitution, the analysis method works for the quantum walk on the half line as well. 
We will see that fact in more detail later.
Also, taking a good look at the form of the operator $C$ in Eq. \eqref{eq:coin-flip_operator}, we realize that all the coefficients of the inner states $\ket{0}$ and $\ket{1}$ should stay in the set of real numbers at any time $t$ as long as the walker starts off with the initial state in Eq. \eqref{eq:L_initial_state}.

The finding probability of the walker at time $t$ is given in the same definition for the quantum walk on the half line,
\begin{equation}
 \mathbb{P}(Y_t^L=x)=\bra{\Phi_t}\biggl(\ket{x}\bra{x}\otimes\sum_{j=0}^1\ket{j}\bra{j}\biggr)\ket{\Phi_t},\label{eq:L_probability}
\end{equation}
where $Y_t^L$ means the position of the quantum walker on the line at time $t$.
Figure \ref{fig:160201_10} demonstrates the probability distribution $\mathbb{P}(Y_t^L=x)$.
Viewing Figs. \ref{fig:160201_10}--(a) and (b), we understand that the probability distribution contains two sharp peaks.
We know in Fig. \ref{fig:160201_10}--(c) how the probability distribution is dependent on the parameter $\theta$ of the operator $C$. 
\begin{figure}[h]
\begin{center}
 \begin{minipage}[b]{35mm}
  \begin{center}
   \includegraphics[scale=0.25]{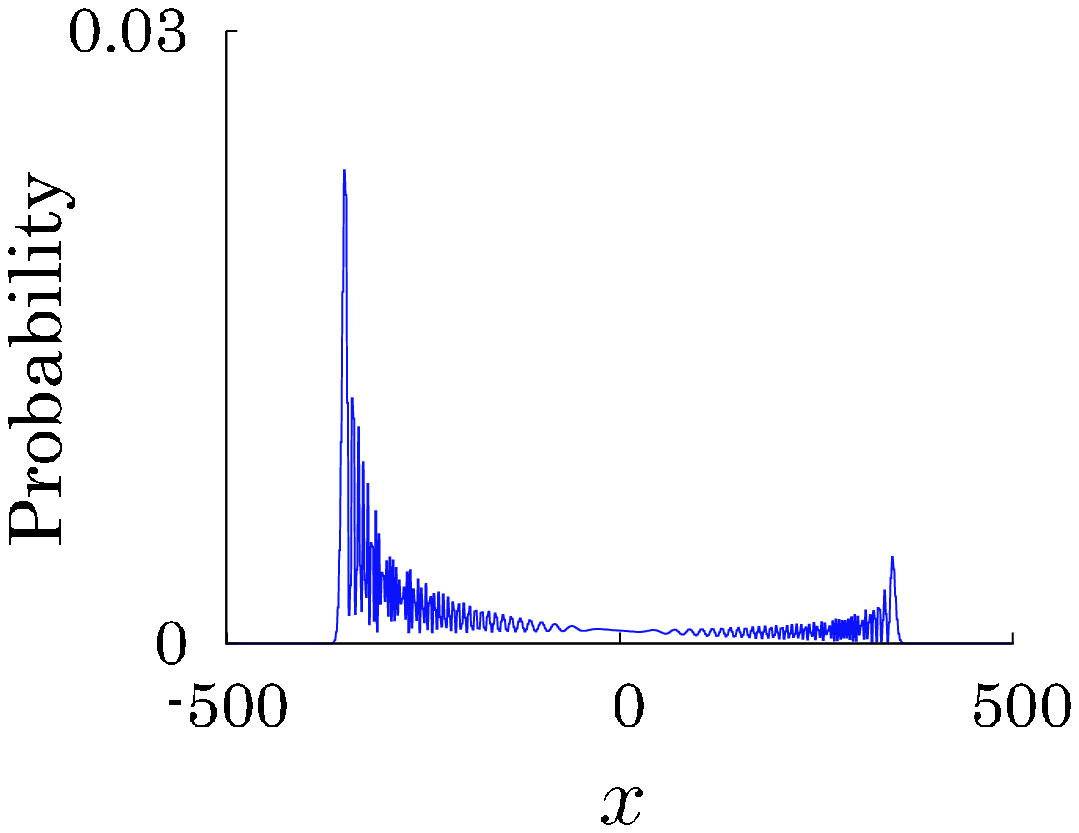}\\[2mm]
  (a)
  \end{center}
 \end{minipage}
 \begin{minipage}[b]{35mm}
  \begin{center}
   \includegraphics[scale=0.15]{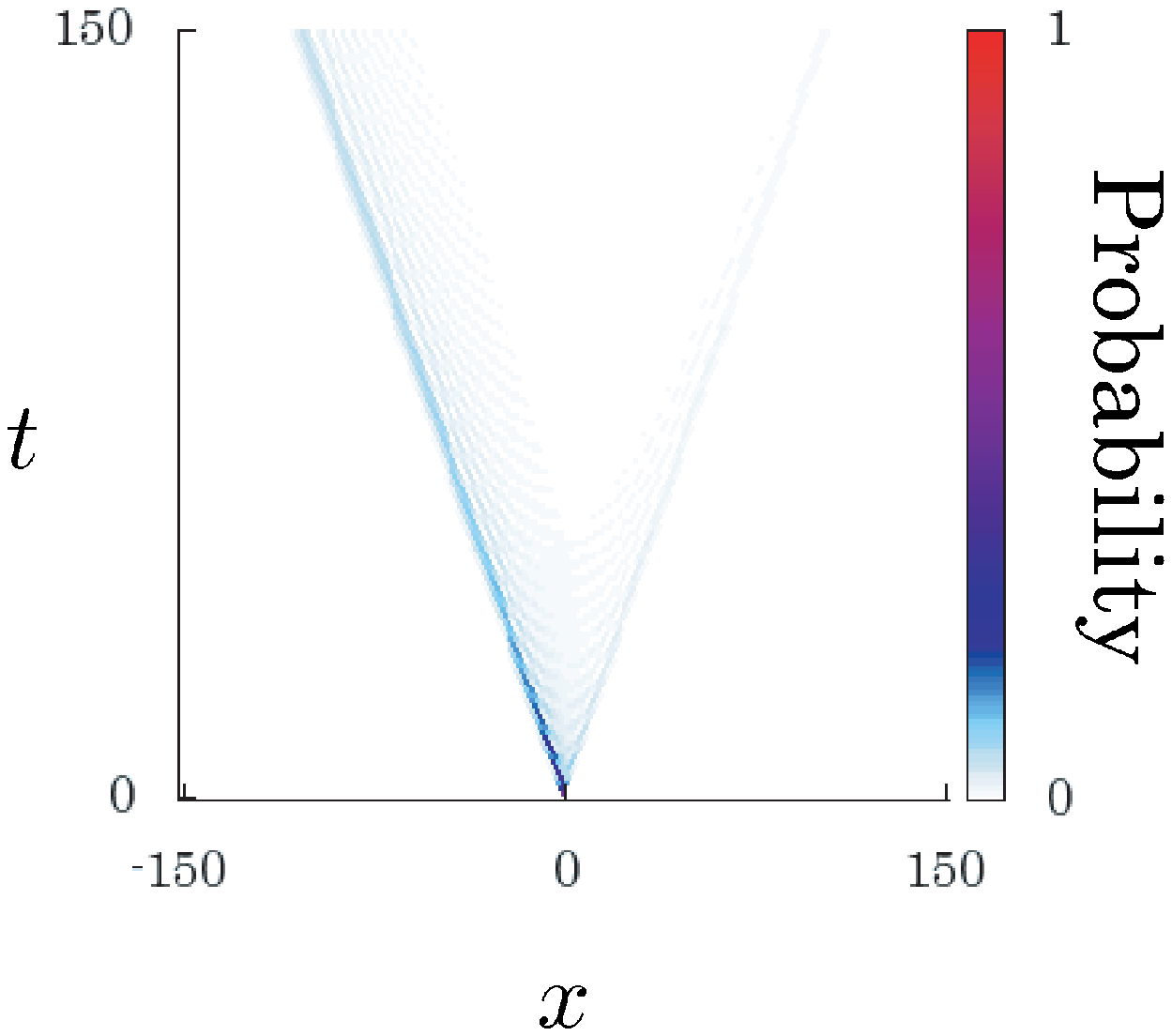}\\[2mm]
  (b)
  \end{center}
 \end{minipage}
 \begin{minipage}[b]{35mm}
  \begin{center}
   \includegraphics[scale=0.15]{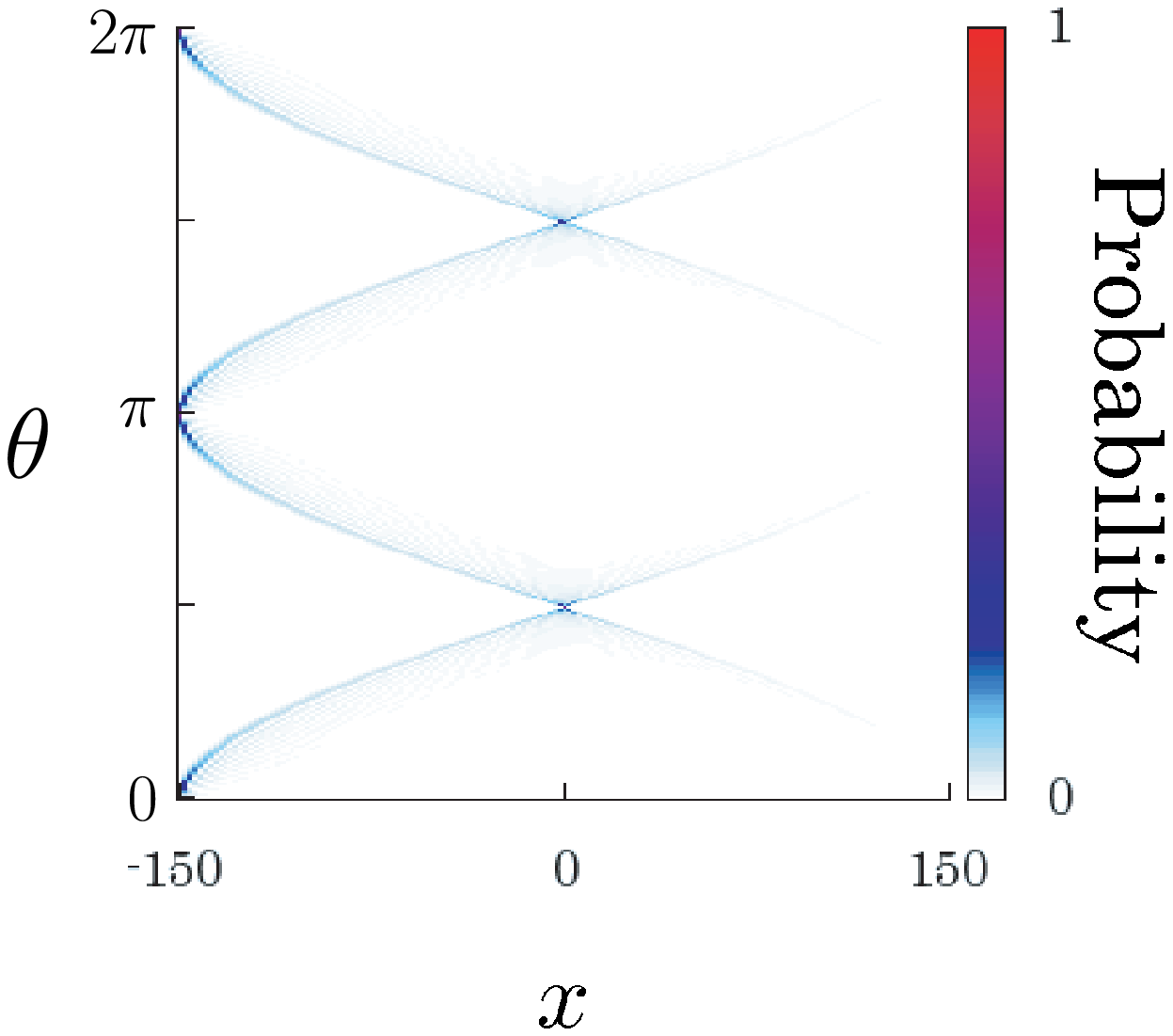}\\[2mm]
  (c)
  \end{center}
 \end{minipage}
\vspace{5mm}
\caption{These pictures show the probability distribution of the quantum walk on the line. (a) $\theta=\pi/4$ : Probability distribution $\mathbb{P}(Y_{500}^L=x)$, (b) $\theta=\pi/4$ : Time evolution of the probability distribution $\mathbb{P}(Y_t^L=x)$, (c) Dependency of the probability distribution $\mathbb{P}(Y_{150}^L=x)$ on the parameter $\theta$.}
\label{fig:160201_10}
\end{center}
\end{figure}

Applying a result in the past study \cite{Konno2002a} to the quantum walk on the line, we are able to encounter an exact representation of the probability distribution.
Assuming that $\theta\neq 0,\pi/2, \pi, 3\pi/2$, one can describe all the positive values of the probability as follows.
\begin{align}
 &\mathbb{P}(Y_t^L=t-2m)=\mathbb{P}(Y_t^L=t-2m-1)\nonumber\\
 =&\frac{c^{2(t-1)}}{2}\sum_{j_1=1}^m \sum_{j_2=1}^m \left(-\frac{s^2}{c^2}\right)^{j_1+j_2} {m-1\choose j_1-1}{m-1\choose j_2-1}{t-m-1\choose j_1-1}{t-m-1\choose j_2-1}\nonumber\\
 &\qquad\qquad\qquad\times\left\{\frac{(m-j_1-j_2)m}{j_1j_2}+\frac{1}{s^2}\right\}\quad (m=1,2,\ldots, [t/2]),\label{eq:151209_5-1}\\[2mm]
 &\mathbb{P}(Y_t^L=-(t-2m)-1)=\mathbb{P}(Y_t^L=-(t-2m))\nonumber\\
 =&\frac{c^{2(t-1)}}{2}\sum_{j_1=1}^m \sum_{j_2=1}^m \left(-\frac{s^2}{c^2}\right)^{j_1+j_2} {m-1\choose j_1-1}{m-1\choose j_2-1}{t-m-1\choose j_1-1}{t-m-1\choose j_2-1}\nonumber\\
 &\qquad\qquad\qquad\times\left\{\frac{(t-m-j_1-j_2)(t-m)}{j_1j_2}+\frac{1}{s^2}\right\}\quad (m=1,2,\ldots, [t/2]),\label{eq:151209_5-2}\\[2mm]
 &\mathbb{P}(Y_t^L=-t-1)=\mathbb{P}(Y_t^L=-t)=\frac{c^{2(t-1)}}{2},\label{eq:151209_5-3}
\end{align}
where we have denoted the floor function by $[\,\cdot\,]$ in Eqs. \eqref{eq:151209_5-1} and \eqref{eq:151209_5-2}, that is, $[x]=\max\left\{n\in\mathbb{Z}\, |\, n\leq x\right\}$ for a real number $x$.

In addition, since the dynamics in Eq. \eqref{eq:L_time-evolution} is spatially homogeneous, one can compute a limit distribution of this quantum walk with the delocalized initial state in Eq. \eqref{eq:L_initial_state} by Fourier analysis as well as the walk on the line with a localized initial state.
Slightly modifying the computation shown in Ref. \cite{GrimmettJansonScudo2004}, we reach a limit distribution
\begin{equation}
 \lim_{t\to\infty}\mathbb{P}\left(\frac{Y_t^L}{t}\leq x\right)=\int_{-\infty}^x \frac{|s|}{\pi(1+y)\sqrt{c^2-y^2}}I_{(-|c|,|c|)}(y)\,dy,\label{eq:L_Fourier_analysis}
\end{equation}
where, for a set $A$, the function $I_A(x)$ represents the indicator function
\begin{equation}
 I_A(x)=\left\{\begin{array}{cl}
	 1&(x\in A) \\
		0&(x\notin A)
	       \end{array}\right..
\end{equation}
Note that this limit distribution is good for the parameter $\theta\neq 0,\pi/2,\pi,3\pi/2$.
Since the behavior of the walker becomes trivial if $\theta=0,\pi/2,\pi, 3\pi/2$, we do not analyze the quantum walks defined by the operator $C$ whose parameter $\theta$ takes such values.
\bigskip

From now on, we are seeing how the quantum walk on the line copies the quantum walk on the half line.
Before that, let us express the systems with the complex numbers $\alpha_t(x), \beta_t(x)$, $\gamma_t(x)$, and $\delta_t(x)$,
\begin{align}
 \ket{\Psi_t}=&\sum_{x=0}^\infty \ket{x}\otimes\Bigl(\alpha_t(x)\ket{0}+\beta_t(x)\ket{1}\Bigr),\label{eq:HL_system_w_ab}\\
 \ket{\Phi_t}=&\sum_{x\in\mathbb{Z}} \ket{x}\otimes\Bigl(\gamma_t(x)\ket{0}+\delta_t(x)\ket{1}\Bigr).\label{eq:L_system_w_gd}
\end{align}
Remembering the initial states in Eqs. \eqref{eq:HL_initial_state} and \eqref{eq:L_initial_state}, we should know the complex numbers expressing the initial states of the quantum walks,
\begin{align}
 \alpha_0(x)=&\left\{\begin{array}{cl}
	       e^{-i\theta}/\sqrt{2} & (x=0)\\
		      0 & (x=1,2,\ldots)
		     \end{array}\right.,\label{eq:alpha_0}\\
 \beta_0(x)=&\left\{\begin{array}{cl}
	       ie^{-i\theta}/\sqrt{2} & (x=0)\\
		      0 & (x=1,2,\ldots)
		     \end{array}\right.,\label{eq:beta_0}\\
 \gamma_0(x)=&\left\{\begin{array}{cl}
	       \cos\theta/\sqrt{2} & (x=-1,0)\\
		      0 & (x\neq -1,0)
		     \end{array}\right.,\label{eq:gamma_0}\\
 \delta_0(x)=&\left\{\begin{array}{cl}
	       \sin\theta/\sqrt{2} & (x=-1,0)\\
		      0 & (x\neq -1,0)
		     \end{array}\right..\label{eq:delta_0}
\end{align}
Moreover, recalling that the quantum walk on the half line can be described by only real numbers in this paper because the initial state is of the form in Eq. \eqref{eq:L_initial_state} and the walker is updating with the operator $C$ in Eq. \eqref{eq:coin-flip_operator}, we are allowed to consider the complex numbers $\gamma_t(x)$ and $\delta_t(x)$ to be real numbers.
This fact may leave our minds for now, but should come back when we read the proof of Theorem \ref{th:151208_18}. 

For the expression in Eq. \eqref{eq:L_system_w_gd}, one can tell a lemma which is only for the quantum walk on the line.
\begin{lem}
\label{lem:151208-5}
 Assume that $\theta\neq 0, \pi$.
 For $x=0,1,2,\ldots$, we have
 \begin{align}
  \delta_{2t}(x-1)=&\delta_{2t}(-x),\label{eq:151208_5-a-1}\\
  s\gamma_{2t}(x-1)-c\delta_{2t}(x-1)=&-s\gamma_{2t}(-x-2)+c\delta_{2t}(-x-2),\label{eq:151208_5-a-2}\\[3mm]
  \delta_{2t+1}(x-1)=&-\delta_{2t+1}(-x),\label{eq:151208_5-i-1}\\
  s\gamma_{2t+1}(x-1)-c\delta_{2t+1}(x-1)=&s\gamma_{2t+1}(-x-2)-c\delta_{2t+1}(-x-2).\label{eq:151208_5-i-2}
 \end{align}
\end{lem}

\begin{proof}{%
First, we prove Eqs. \eqref{eq:151208_5-i-1} and \eqref{eq:151208_5-i-2}, assuming Eqs. \eqref{eq:151208_5-a-1} and \eqref{eq:151208_5-a-2} which will be used for Eqs. \eqref{eq:160205_1} and \eqref{eq:160205_2}.
For $x=0,1,2,\ldots$, Eq. \eqref{eq:L_time-evolution} gives
\begin{align}
 \delta_{2t+1}(x-1)=&s\gamma_{2t}(x-2)-c\delta_{2t}(x-2),\\[2mm]
 \delta_{2t+1}(-x)=&s\gamma_{2t}(-x-1)-c\delta_{2t}(-x-1)\nonumber\\
 =&-s\gamma_{2t}(x-2)+c\delta_{2t}(x-2)\label{eq:160205_1},\\[2mm]
  s\gamma_{2t+1}(x-1)-c\delta_{2t+1}(x-1)=& cs\gamma_{2t}(x)+s^2\delta_{2t}(x)\nonumber\\
 &-cs\gamma_{2t}(x-2)+c^2\delta_{2t}(x-2),\\[2mm]
  s\gamma_{2t+1}(-x-2)-c\delta_{2t+1}(-x-2)=&c\left\{s\gamma_{2t}(-x-1)-c\delta_{2t}(-x-1)\right\}\nonumber\\
 &+\delta_{2t}(-x-1)\nonumber\\
 &-c\left\{s\gamma_{2t}(-x-3)-c\delta_{2t}(-x-3)\right\}\nonumber\\
 =& cs\gamma_{2t}(x)+s^2\delta_{2t}(x)\nonumber\\
 &-cs\gamma_{2t}(x-2)+c^2\delta_{2t}(x-2)\label{eq:160205_2},
\end{align}
from which Eqs. \eqref{eq:151208_5-i-1} and \eqref{eq:151208_5-i-2} follow.
Next let us aim at Eqs. \eqref{eq:151208_5-a-1} and \eqref{eq:151208_5-a-2} where the subscript $t$ is replaced with $t+1$, assuming Eqs. \eqref{eq:151208_5-i-1} and \eqref{eq:151208_5-i-2} which will be used for Eqs. \eqref{eq:160205_3} and \eqref{eq:160205_4}.
For $x=0,1,2,\ldots$, Eq. \eqref{eq:L_time-evolution} gives
\begin{align}
 \delta_{2t+2}(x-1)=&s\gamma_{2t+1}(x-2)-c\delta_{2t+1}(x-2),\\[2mm]
 \delta_{2t+2}(-x)=&s\gamma_{2t+1}(-x-1)-c\delta_{2t+1}(-x-1)\nonumber\\
 =&s\gamma_{2t+1}(x-2)-c\delta_{2t+1}(x-2),\label{eq:160205_3}\\[2mm]
 s\gamma_{2t+2}(x-1)-c\delta_{2t+2}(x-1)=&cs\gamma_{2t+1}(x)+s^2\delta_{2t+1}(x)\nonumber\\
 &-cs\gamma_{2t+1}(x-2)+c^2\delta_{2t+1}(x-2),\\[2mm]
 -s\gamma_{2t+2}(-x-2)+c\delta_{2t+2}(-x-2)=&-c\left\{s\gamma_{2t+1}(-x-1)-c\delta_{2t+1}(-x-1)\right\}\nonumber\\
 &-\delta_{2t+1}(-x-1)\nonumber\\
 &+c\left\{s\gamma_{2t+1}(-x-3)-c\delta_{2t+1}(-x-3)\right\}\nonumber\\
 =&cs\gamma_{2t+1}(x)+s^2\delta_{2t+1}(x)\nonumber\\
 &-cs\gamma_{2t+1}(x-2)+c^2\delta_{2t+1}(x-2).\label{eq:160205_4}
\end{align}
The equations which we aimed at, hence, come up.
Since Eqs. \eqref{eq:gamma_0} and \eqref{eq:delta_0} mean that Eqs. \eqref{eq:151208_5-a-1} and \eqref{eq:151208_5-a-2} are true as $t=0$, one can assert Lemma \ref{lem:151208-5} by mathematical induction.
\qed
}
\end{proof}
\bigskip

The next lemma finally tells us the fact that the quantum walk on the line reproduces the quantum walk on the half line.
\begin{lem}
\label{lem:151208-6}
 Assume that $\theta\neq 0, \pi$.
 For $x=0,1,2,\ldots$, we have
 \begin{align}
  \alpha_{2t}(2x)=&\gamma_{2t}(2x)-i\delta_{2t}(2x),\label{eq:151208_6-a-1}\\
  \beta_{2t}(2x)=&\delta_{2t}(-2x-1)+i\gamma_{2t}(-2x-1),\label{eq:151208_6-a-2}\\
  \alpha_{2t}(2x+1)=&\delta_{2t}(2x+1)+i\gamma_{2t}(2x+1),\label{eq:151208_6-a-3}\\
  \beta_{2t}(2x+1)=&-\gamma_{2t}(-2x-2)+i\delta_{2t}(-2x-2),\label{eq:151208_6-a-4}\\[2mm]
  \alpha_{2t+1}(2x)=&\delta_{2t+1}(2x)+i\gamma_{2t+1}(2x),\label{eq:151208_6-i-1}\\
  \beta_{2t+1}(2x)=&\gamma_{2t+1}(-2x-1)-i\delta_{2t+1}(-2x-1),\label{eq:151208_6-i-2}\\
  \alpha_{2t+1}(2x+1)=&\gamma_{2t+1}(2x+1)-i\delta_{2t+1}(2x+1),\label{eq:151208_6-i-3}\\
  \beta_{2t+1}(2x+1)=&-\delta_{2t+1}(-2x-2)-i\gamma_{2t+1}(-2x-2).\label{eq:151208_6-i-4}
 \end{align}
\end{lem}

\begin{proof}{%
This lemma can be proved by mathematical induction.
Equations \eqref{eq:alpha_0}--\eqref{eq:delta_0} guarantee Eqs. \eqref{eq:151208_6-a-1}--\eqref{eq:151208_6-a-4} as $t=0$.
Let us assume Eqs. \eqref{eq:151208_6-a-1}--\eqref{eq:151208_6-a-4}.
Under the assumption, Eq. \eqref{eq:HL_time-evolution} expresses $\alpha_{2t+1}(x)$ and $\beta_{2t+1}(x)$ in terms of $\gamma_{2t}(x)$ and $\delta_{2t}(x)$.
Due to the spatial inhomogeneity of the operator $\tilde{S}^{HL}$ in Eq. \eqref{eq:HL_shift_operator}, we should take care of the inner states at the origin separately from the other positions,
\begin{align}
 \alpha_{2t+1}(0)=&c\alpha_{2t}(1)+s\beta_{2t}(1)=-s\gamma_{2t}(-2)+c\delta_{2t}(1)+i\left\{c\gamma_{2t}(1)+s\delta_{2t}(-2)\right\},\label{eq:160207-1}\\
 \beta_{2t+1}(0)=&c\alpha_{2t}(0)+s\beta_{2t}(0)=c\gamma_{2t}(0)+s\delta_{2t}(-1)+i\left\{s\gamma_{2t}(-1)-c\delta_{2t}(0)\right\}.\label{eq:160207-2}
\end{align}
On the other hand, for $x=1,2,\ldots$, we have
\begin{align}
 \alpha_{2t+1}(2x)=&c\alpha_{2t}(2x+1)+s\beta_{2t}(2x+1)\nonumber\\
 =&-s\gamma_{2t}(-2x-2)+c\delta_{2t}(2x+1)\nonumber\\
 &+i\left\{c\gamma_{2t}(2x+1)+s\delta_{2t}(-2x-2)\right\},\label{eq:160207-3}\\[2mm]
 \beta_{2t+1}(2x)=&s\alpha_{2t}(2x-1)-c\beta_{2t}(2x-1)\nonumber\\
 =&c\gamma_{2t}(-2x)+s\delta_{2t}(2x-1)+i\left\{s\gamma_{2t}(2x-1)-c\delta_{2t}(-2x)\right\}.\label{eq:160207-4}
\end{align}
As for Eqs. \eqref{eq:151208_6-a-3} and  \eqref{eq:151208_6-a-4}, for $x=0,1,2,\ldots$, we have
\begin{align}
 \alpha_{2t+1}(2x+1)=&c\alpha_{2t}(2x+2)+s\beta_{2t}(2x+2)\nonumber\\
 =&c\gamma_{2t}(2x+2)+s\delta_{2t}(-2x-3)\nonumber\\
 &+i\left\{s\gamma_{2t}(-2x-3)-c\delta_{2t}(2x+2)\right\},\label{eq:160207-5}\\[2mm]
 \beta_{2t+1}(2x+1)=&s\alpha_{2t}(2x)-c\beta_{2t}(2x)\nonumber\\
 =&s\gamma_{2t}(2x)-c\delta_{2t}(-2x-1)-i\left\{c\gamma_{2t}(-2x-1)+s\delta_{2t}(2x)\right\}.\label{eq:160207-6}
\end{align}
The usage of Lemma \ref{lem:151208-5} and Eq. \eqref{eq:L_time-evolution} summons $\gamma_{2t+1}(x)$ and $\delta_{2t+1}(x)$ in Eqs. \eqref{eq:160207-1}--\eqref{eq:160207-6},
\begin{align}
 \alpha_{2t+1}(0)=&s\gamma_{2t}(-1)-c\delta_{2t}(-1)+i\left\{c\gamma_{2t}(1)+s\delta_{2t}(1)\right\}\nonumber\\
 =&\delta_{2t+1}(0)+i\gamma_{2t+1}(0),\\[2mm]
 \beta_{2t+1}(0)=&c\gamma_{2t}(0)+s\delta_{2t}(0)+i\left\{-s\gamma_{2t}(-2)+c\delta_{2t}(-2)\right\}\nonumber\\
 =&\gamma_{2t+1}(-1)-i\delta_{2t+1}(-1),
\end{align}
\begin{align}
 \alpha_{2t+1}(2x)=&s\gamma_{2t}(2x-1)-c\delta_{2t}(2x-1)+i\left\{c\gamma_{2t}(2x+1)+s\delta_{2t}(2x+1)\right\}\nonumber\\
 =&\delta_{2t+1}(2x)+i\gamma_{2t+1}(2x)\qquad (x=1,2,\ldots),\\[2mm]
 \beta_{2t+1}(2x)=&c\gamma_{2t}(-2x)+s\delta_{2t}(-2x)+i\left\{-s\gamma_{2t}(-2x-2)+c\delta_{2t}(-2x-2)\right\}\nonumber\\
 =&\gamma_{2t+1}(-2x-1)-i\delta_{2t+1}(-2x-1)\qquad (x=1,2,\ldots),
\end{align}
\begin{align}
 \alpha_{2t+1}(2x+1)=&c\gamma_{2t}(2x+2)+s\delta_{2t}(2x+2)+i\left\{-s\gamma_{2t}(2x)+c\delta_{2t}(2x)\right\}\nonumber\\
 =&\gamma_{2t+1}(2x+1)-i\delta_{2t+1}(2x+1)\qquad (x=0,1,2,\ldots),\\[2mm]
 \beta_{2t+1}(2x+1)=&-s\gamma_{2t}(-2x-3)+c\delta_{2t}(-2x-3)\nonumber\\
 &-i\left\{c\gamma_{2t}(-2x-1)+s\delta_{2t}(-2x-1)\right\}\nonumber\\
 =&-\delta_{2t+1}(-2x-2)-i\gamma_{2t+1}(-2x-2)\qquad (x=0,1,2,\ldots),
\end{align}
which are all combined as Eqs. \eqref{eq:151208_6-i-1}--\eqref{eq:151208_6-i-4}.

Next, assuming Eqs. \eqref{eq:151208_6-i-1}--\eqref{eq:151208_6-i-4}, we derive Eqs. \eqref{eq:151208_6-a-1}--\eqref{eq:151208_6-a-4} in which the subscript $t$ is replaced with $t+1$.
We take the assumption and repeat a similar computation using Eq. \eqref{eq:HL_time-evolution}, Eq. \eqref{eq:L_time-evolution}, and Lemma \ref{lem:151208-5},
\begin{align}
 \alpha_{2t+2}(0)=&c\alpha_{2t+1}(1)+s\beta_{2t+1}(1)\nonumber\\
 =&c\gamma_{2t+1}(1)-s\delta_{2t+1}(-2)-i\left\{s\gamma_{2t+1}(-2)+c\delta_{2t+1}(1)\right\}\nonumber\\
 =&c\gamma_{2t+1}(1)+s\delta_{2t+1}(1)-i\left\{s\gamma_{2t+1}(-1)-c\delta_{2t+1}(-1)\right\}\nonumber\\
 =&\gamma_{2t+2}(0)-i\delta_{2t+2}(0),\\[2mm]
 \beta_{2t+2}(0)=&c\alpha_{2t+1}(0)+s\beta_{2t+1}(0)\nonumber\\
 =&s\gamma_{2t+1}(-1)+c\delta_{2t+1}(0)+i\left\{c\gamma_{2t+1}(0)-s\delta_{2t+1}(-1)\right\}\nonumber\\
 =&s\gamma_{2t+1}(-2)-c\delta_{2t+1}(-2)+i\left\{c\gamma_{2t+1}(0)+s\delta_{2t+1}(0)\right\}\nonumber\\
 =&\delta_{2t+2}(-1)+i\gamma_{2t+2}(-1),
\end{align}
\begin{align}
 \alpha_{2t+2}(2x)=&c\alpha_{2t+1}(2x+1)+s\beta_{2t+1}(2x+1)\nonumber\\
 =&c\gamma_{2t+1}(2x+1)-s\delta_{2t+1}(-2x-2)\nonumber\\
 &-i\left\{s\gamma_{2t+1}(-2x-2)+c\delta_{2t+1}(2x+1)\right\}\nonumber\\
 =&c\gamma_{2t+1}(2x+1)+s\delta_{2t+1}(2x+1)\nonumber\\
 &-i\left\{s\gamma_{2t+1}(2x-1)-c\delta_{2t+1}(2x-1)\right\}\nonumber\\
 =&\gamma_{2t+2}(2x)-i\delta_{2t+2}(2x)\qquad (x=1,2,\ldots),\\[2mm]
 \beta_{2t+2}(2x)=&s\alpha_{2t+1}(2x-1)-c\beta_{2t+1}(2x-1)\nonumber\\
 =&s\gamma_{2t+1}(2x-1)+c\delta_{2t+1}(-2x)\nonumber\\
 &+i\left\{c\gamma_{2t+1}(-2x)-s\delta_{2t+1}(2x-1)\right\}\nonumber\\
 =&s\gamma_{2t+1}(-2x-2)-c\delta_{2t+1}(-2x-2)\nonumber\\
 &+i\left\{c\gamma_{2t+1}(-2x)+s\delta_{2t+1}(-2x)\right\}\nonumber\\
 =&\delta_{2t+2}(-2x-1)+i\gamma_{2t+2}(-2x-1)\qquad (x=1,2,\ldots),
\end{align}
\begin{align}
 \alpha_{2t+2}(2x+1)=&c\alpha_{2t+1}(2x+2)+s\beta_{2t+1}(2x+2)\nonumber\\
 =&s\gamma_{2t+1}(-2x-3)+c\delta_{2t+1}(2x+2)\nonumber\\
 &+i\left\{c\gamma_{2t+1}(2x+2)-s\delta_{2t+1}(-2x-3)\right\}\nonumber\\
 =&s\gamma_{2t+1}(2x)-c\delta_{2t+1}(2x)\nonumber\\
 &+i\left\{c\gamma_{2t+1}(2x+2)+s\delta_{2t+1}(2x+2)\right\}\nonumber\\
 =&\delta_{2t+2}(2x+1)+i\gamma_{2t+2}(2x+1)\qquad (x=0,1,2,\ldots),\\[2mm]
 \beta_{2t+2}(2x+1)=&s\alpha_{2t+1}(2x)-c\beta_{2t+1}(2x)\nonumber\\
 =&-c\gamma_{2t+1}(-2x-1)+s\delta_{2t+1}(2x)\nonumber\\
 &+i\left\{s\gamma_{2t+1}(2x)+c\delta_{2t+1}(-2x-1)\right\}\nonumber\\
 =&-c\gamma_{2t+1}(-2x-1)-s\delta_{2t+1}(-2x-1)\nonumber\\
 &+i\left\{s\gamma_{2t+1}(-2x-3)-c\delta_{2t+1}(-2x-3)\right\}\nonumber\\
 =&-\gamma_{2t+2}(-2x-2)+i\delta_{2t+2}(-2x-2)\qquad (x=0,1,2,\ldots).
\end{align}
Consequently, for $x=0,1,2,\ldots$, a bunch of equations have been figured out,
\begin{align}
 \alpha_{2t+2}(2x)=&\gamma_{2t+2}(2x)-i\delta_{2t+2}(2x),\\
 \beta_{2t+2}(2x)=&\delta_{2t+2}(-2x-1)+i\gamma_{2t+2}(-2x-1),\\
 \alpha_{2t+2}(2x+1)=&\delta_{2t+2}(2x+1)+i\gamma_{2t+2}(2x+1),\\
 \beta_{2t+2}(2x+1)=&-\gamma_{2t+2}(-2x-2)+i\delta_{2t+2}(-2x-2),
\end{align}
and they are the things which were wanted under the assumption Eqs. \eqref{eq:151208_6-i-1}--\eqref{eq:151208_6-i-4}.
Based on all the fact shown here, one can tell Lemma \ref{lem:151208-5} by mathematical induction.
\qed
}
\end{proof}
\bigskip

Now that the relation between two quantum walks has been discovered, it gives a connection between the finding probabilities.
\begin{thm}
\label{th:151208_18}
 Assume that $\theta\neq 0, \pi$.
 For $x=0,1,2,\ldots$, we have
 \begin{align}
  \mathbb{P}(X_t^{HL}=x;0)=&\mathbb{P}(Y_t^L=x),\\
  \mathbb{P}(X_t^{HL}=x;1)=&\mathbb{P}(Y_t^L=-x-1),\\
  \mathbb{P}(X_t^{HL}=x)=&\mathbb{P}(Y_t^L=-x-1)+\mathbb{P}(Y_t^L=x).
 \end{align}
\end{thm}

\begin{proof}{%
First, let us recall the finding probabilities in Eqs. \eqref{eq:probability_inner_state}, \eqref{eq:HL_probability}, and \eqref{eq:L_probability}, and the expressions in Eqs. \eqref{eq:HL_system_w_ab} and \eqref{eq:L_system_w_gd}.
Then we get the representation of each probability by the complex numbers $\alpha_t(x), \beta_t(x), \gamma_t(x)$, and $\delta_t(x)$,
\begin{align}
 \mathbb{P}(X_t^{HL}=x;0)=&|\alpha_t(x)|^2,\\
 \mathbb{P}(X_t^{HL}=x;1)=&|\beta_t(x)|^2,\\
 \mathbb{P}(Y_t^L=x)=&|\gamma_t(x)|^2+|\delta_t(x)|^2.
\end{align}
On the other hand, since the complex numbers $\gamma_t(x)$ and $\delta_t(x)$ always stay in the set of real numbers in this paper, Lemma \ref{lem:151208-6} shows
\begin{align}
 |\alpha_t(x)|^2=&\gamma_t(x)^2+\delta_t(x)^2=|\gamma_t(x)|^2+|\delta_t(x)|^2,\\
 |\beta_t(x)|^2=&\gamma_t(-x-1)^2+\delta_t(-x-1)^2=|\gamma_t(-x-1)|^2+|\delta_t(-x-1)|^2,
\end{align}
from which we find the statement of Theorem \ref{th:151208_18}.
\qed
}
\end{proof}
\bigskip

Combining Eqs. \eqref{eq:151209_5-1}, \eqref{eq:151209_5-2}, and \eqref{eq:151209_5-3} with Theorem \ref{th:151208_18}, we see a representation of the probability distribution $\mathbb{P}(X_t^{HL}=x;j)$.

\begin{lem}
\label{lem:151209_7p}
 Assume that $\theta\neq 0,\pi/2, \pi, 3\pi/2$.
 One can describe all the positive values of the probability $\mathbb{P}(X_t^{HL}=x;j)\,(j\in\left\{0,1\right\})$ as follows.
 For $t=1,2,\ldots$, we have
 \begin{align}
  &\mathbb{P}(X_{2t}^{HL}=2(t-m);0)=\mathbb{P}(X_{2t}^{HL}=2(t-m)-1;0)\nonumber\\
  =&\frac{c^{2(2t-1)}}{2}\sum_{j_1=1}^m \sum_{j_2=1}^m \left(-\frac{s^2}{c^2}\right)^{j_1+j_2} {m-1\choose j_1-1}{m-1\choose j_2-1}{2t-m-1\choose j_1-1}\nonumber\\
  \times & {2t-m-1\choose j_2-1}\left\{\frac{(m-j_1-j_2)m}{j_1j_2}+\frac{1}{s^2}\right\}\quad (m=1,2,\ldots, t-1),\\[2mm]
  &\mathbb{P}(X_{2t}^{HL}=2(t-m);1)=\mathbb{P}(X_{2t}^{HL}=2(t-m)-1;1)\nonumber\\
  =&\frac{c^{2(2t-1)}}{2}\sum_{j_1=1}^m \sum_{j_2=1}^m \left(-\frac{s^2}{c^2}\right)^{j_1+j_2} {m-1\choose j_1-1}{m-1\choose j_2-1}{2t-m-1\choose j_1-1}\nonumber\\
  \times & {2t-m-1\choose j_2-1}\left\{\frac{(2t-m-j_1-j_2)(2t-m)}{j_1j_2}+\frac{1}{s^2}\right\}\quad (m=1,2,\ldots, t-1),\\[2mm]
  &\mathbb{P}(X_{2t}^{HL}=0;0)=\mathbb{P}(X_{2t}^{HL}=0;1)\nonumber\\
  =&\frac{c^{2(2t-1)}}{2}\sum_{j_1=1}^t \sum_{j_2=1}^t \left(-\frac{s^2}{c^2}\right)^{j_1+j_2} {t-1\choose j_1-1}^2 {t-1\choose j_2-1}^2 \left\{\frac{(t-j_1-j_2)t}{j_1j_2}+\frac{1}{s^2}\right\},\\[2mm]
  &\mathbb{P}(X_{2t}^{HL}=2t;1)=\mathbb{P}(X_{2t}^{HL}=2t-1;1)=\frac{c^{2(2t-1)}}{2}.
 \end{align}
 For $t=0,1,2,\ldots$, we have
 \begin{align}
  &\mathbb{P}(X_{2t+1}^{HL}=2(t-m)+1;0)=\mathbb{P}(X_{2t+1}^{HL}=2(t-m);0)\nonumber\\
  =&\frac{c^{4t}}{2}\sum_{j_1=1}^m \sum_{j_2=1}^m \left(-\frac{s^2}{c^2}\right)^{j_1+j_2} {m-1\choose j_1-1}{m-1\choose j_2-1}{2t-m\choose j_1-1}{2t-m\choose j_2-1}\nonumber\\
  &\qquad\qquad\qquad\times\left\{\frac{(m-j_1-j_2)m}{j_1j_2}+\frac{1}{s^2}\right\}\quad (m=1,2,\ldots, t),\\[2mm]
  &\mathbb{P}(X_{2t+1}^{HL}=2(t-m)+1;1)=\mathbb{P}(X_{2t+1}^{HL}=2(t-m);1)\nonumber\\
  =&\frac{c^{4t}}{2}\sum_{j_1=1}^m \sum_{j_2=1}^m \left(-\frac{s^2}{c^2}\right)^{j_1+j_2} {m-1\choose j_1-1}{m-1\choose j_2-1}{2t-m\choose j_1-1}{2t-m\choose j_2-1}\nonumber\\
  &\times\left\{\frac{(2t+1-m-j_1-j_2)(2t+1-m)}{j_1j_2}+\frac{1}{s^2}\right\}\quad (m=1,2,\ldots, t),\\[2mm]
   &\mathbb{P}(X_{2t+1}^{HL}=2t+1;1)=\mathbb{P}(X_{2t+1}^{HL}=2t;1)=\frac{c^{4t}}{2}.
 \end{align}
\end{lem}
\bigskip

From Lemma \ref{lem:151209_7p}, one can obtain a theorem for the probability distribution $\mathbb{P}(X_t^{HL}=x)$.
 
\begin{thm}
\label{th:151209_6}
 Assume that $\theta\neq 0,\pi/2, \pi, 3\pi/2$.
 One can describe all the positive values of the probability $\mathbb{P}(X_t^{HL}=x)$ as follows.
 For $t=1,2,\ldots$, we have
 \begin{align}
  &\mathbb{P}(X_{2t}^{HL}=2(t-m))=\mathbb{P}(X_{2t}^{HL}=2(t-m)-1)\nonumber\\
  =&\frac{c^{2(2t-1)}}{2}\sum_{j_1=1}^m \sum_{j_2=1}^m \left(-\frac{s^2}{c^2}\right)^{j_1+j_2} {m-1\choose j_1-1}{m-1\choose j_2-1}{2t-m-1\choose j_1-1}\nonumber\\
  &\times{2t-m-1\choose j_2-1}\left\{\frac{m^2+(2t-m)^2-2(j_1+j_2)t}{j_1j_2}+\frac{2}{s^2}\right\}\nonumber\\
  &\hspace{7cm} (m=1,2,\ldots, t-1),\\[2mm]
  &\mathbb{P}(X_{2t}^{HL}=0)\nonumber\\
  =&c^{2(2t-1)}\sum_{j_1=1}^t \sum_{j_2=1}^t \left(-\frac{s^2}{c^2}\right)^{j_1+j_2} {t-1\choose j_1-1}^2 {t-1\choose j_2-1}^2 \left\{\frac{(t-j_1-j_2)t}{j_1j_2}+\frac{1}{s^2}\right\},\\[2mm]
  &\mathbb{P}(X_{2t}^{HL}=2t)=\mathbb{P}(X_{2t}^{HL}=2t-1)=\frac{c^{2(2t-1)}}{2}.
 \end{align}
 For $t=0,1,2,\ldots$, we have
 \begin{align}
  &\mathbb{P}(X_{2t+1}^{HL}=2(t-m)+1)=\mathbb{P}(X_{2t+1}^{HL}=2(t-m))\nonumber\\
  =&\frac{c^{4t}}{2}\sum_{j_1=1}^m \sum_{j_2=1}^m \left(-\frac{s^2}{c^2}\right)^{j_1+j_2} {m-1\choose j_1-1}{m-1\choose j_2-1}{2t-m\choose j_1-1}{2t-m\choose j_2-1}\nonumber\\
  &\times\left\{\frac{m^2+(2t+1-m)^2-(j_1+j_2)(2t+1)}{j_1j_2}+\frac{2}{s^2}\right\}\quad (m=1,2,\ldots, t),\\[2mm]
  &\mathbb{P}(X_{2t+1}^{HL}=2t+1)=\mathbb{P}(X_{2t+1}^{HL}=2t)=\frac{c^{4t}}{2}.
 \end{align}
\end{thm}

As shown in Figs. \ref{fig:160201_13}--\ref{fig:160201_22}, Lemma \ref{lem:151209_7p} and Theorem \ref{th:151209_6} are in complete agreement with numerical experiments performed according to Eq. \eqref{eq:L_time-evolution}.
The computation of the blue bars are based on Eq. \eqref{eq:L_time-evolution} and the red filled circles are estimated by the representations in Lemma \ref{lem:151209_7p} and Theorem \ref{th:151209_6}. 

\begin{figure}[h]
\begin{center}
 \begin{minipage}{35mm}
  \begin{center}
   \includegraphics[scale=0.3]{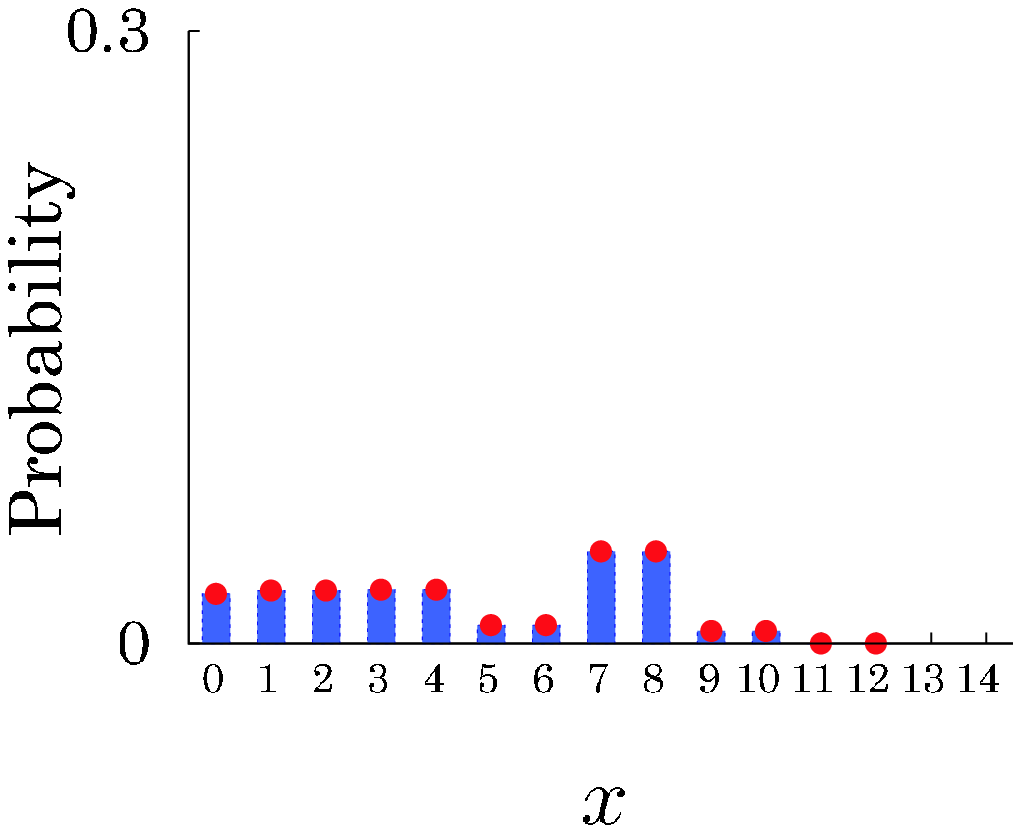}\\[2mm]
  (a) $\mathbb{P}(X_{14}^{HL}=x;0)$
  \end{center}
 \end{minipage}
 \begin{minipage}{35mm}
  \begin{center}
   \includegraphics[scale=0.3]{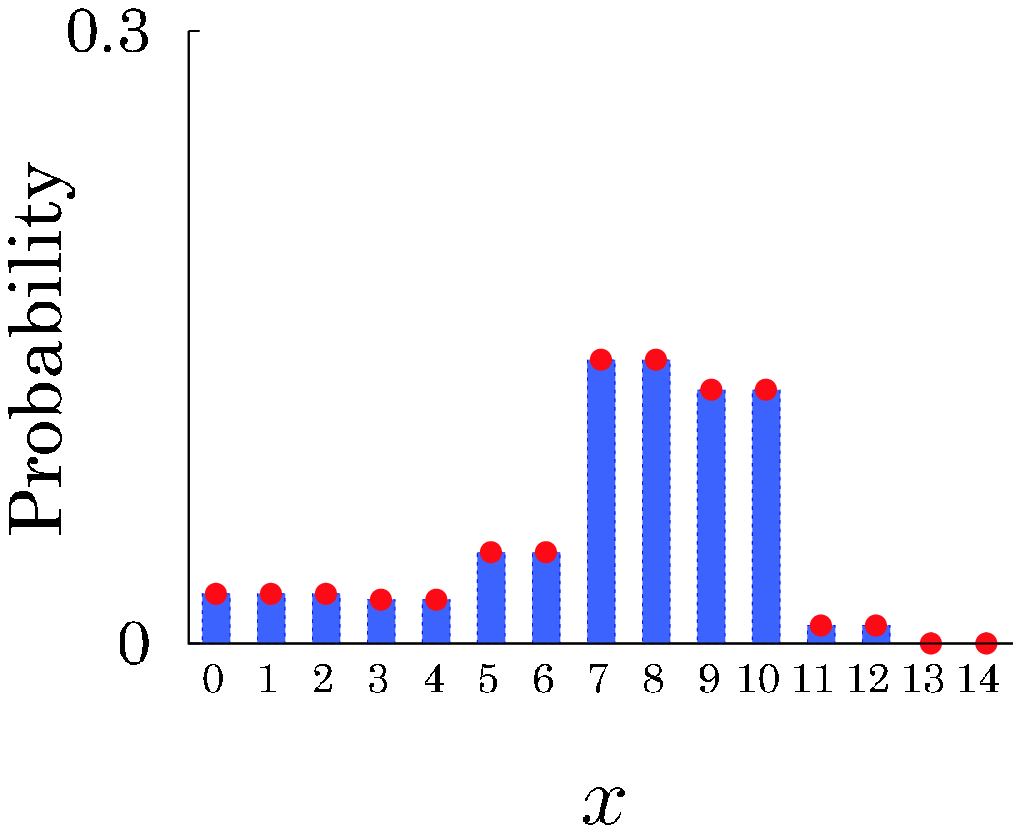}\\[2mm]
  (b) $\mathbb{P}(X_{14}^{HL}=x;1)$
  \end{center}
 \end{minipage}
 \begin{minipage}{35mm}
  \begin{center}
   \includegraphics[scale=0.3]{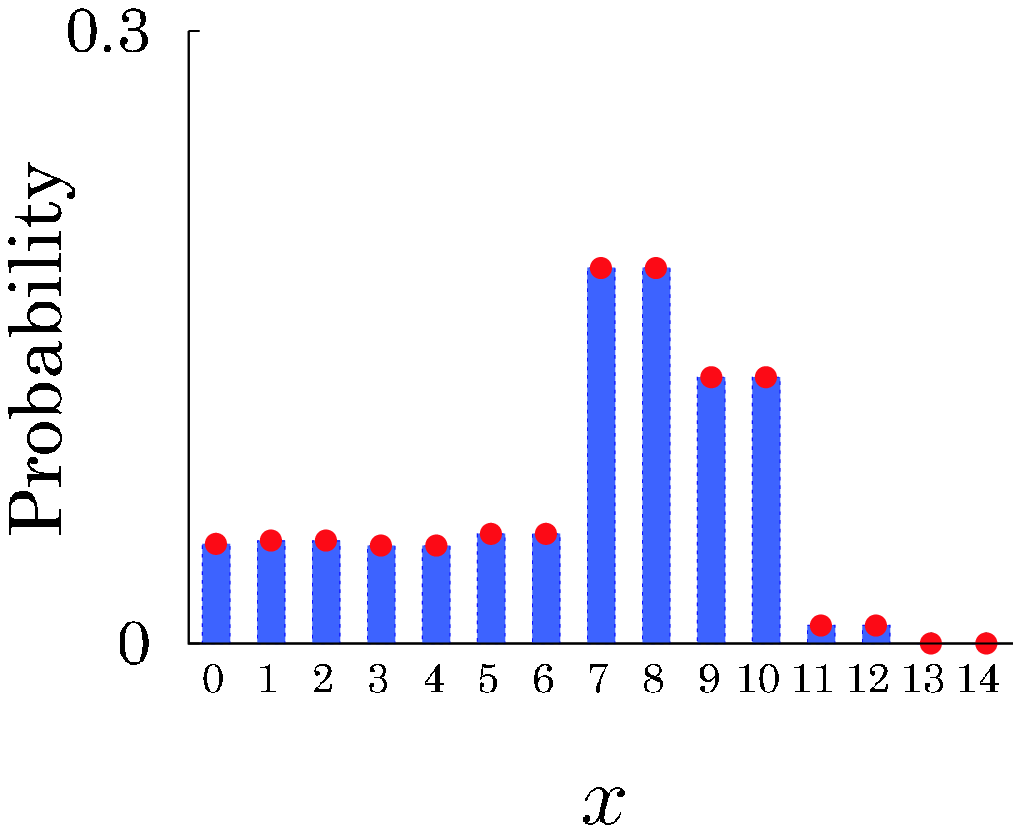}\\[2mm]
  (c) $\mathbb{P}(X_{14}^{HL}=x)$
  \end{center}
 \end{minipage}
\vspace{5mm}
\caption{$\theta=\pi/4$ : The results of the computation according to Eq. \eqref{eq:L_time-evolution} are represented by the blue bars. The red filled circles are estimated by Lemma \ref{lem:151209_7p} and Theorem \ref{th:151209_6} as $t=14$ which is an even number.}
\label{fig:160201_13}
\end{center}
\end{figure}

\begin{figure}[h]
\begin{center}
 \begin{minipage}{35mm}
  \begin{center}
   \includegraphics[scale=0.3]{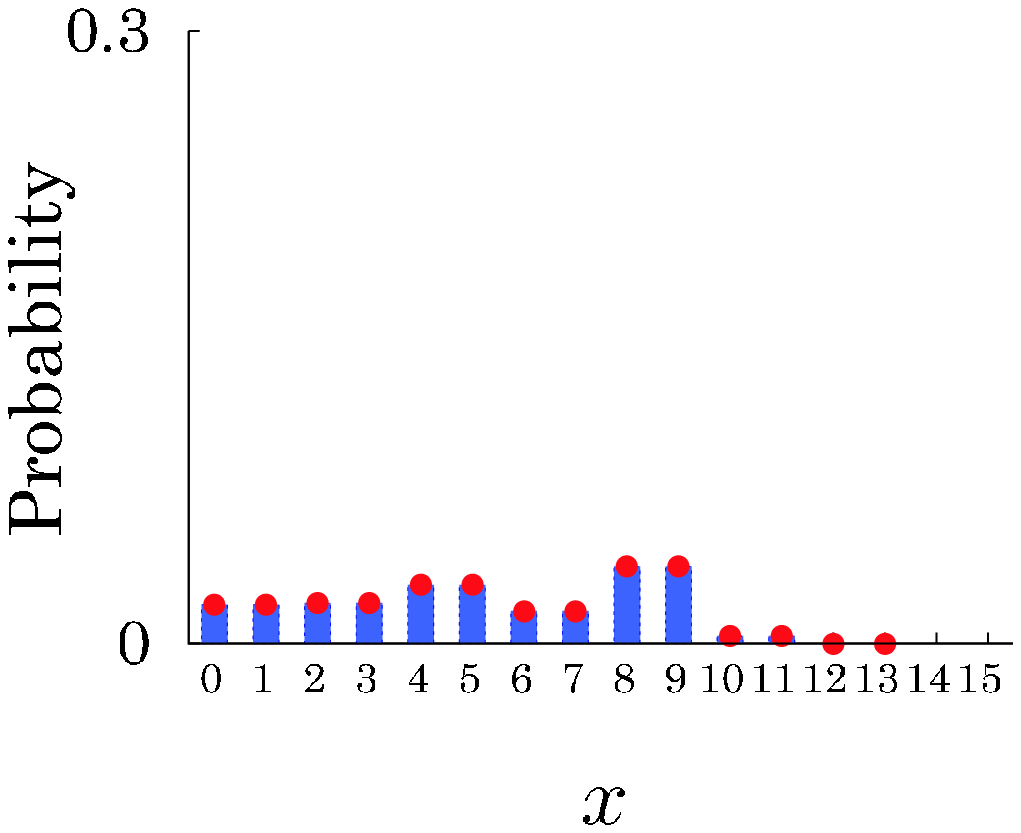}\\[2mm]
  (a) $\mathbb{P}(X_{15}^{HL}=x;0)$
  \end{center}
 \end{minipage}
 \begin{minipage}{35mm}
  \begin{center}
   \includegraphics[scale=0.3]{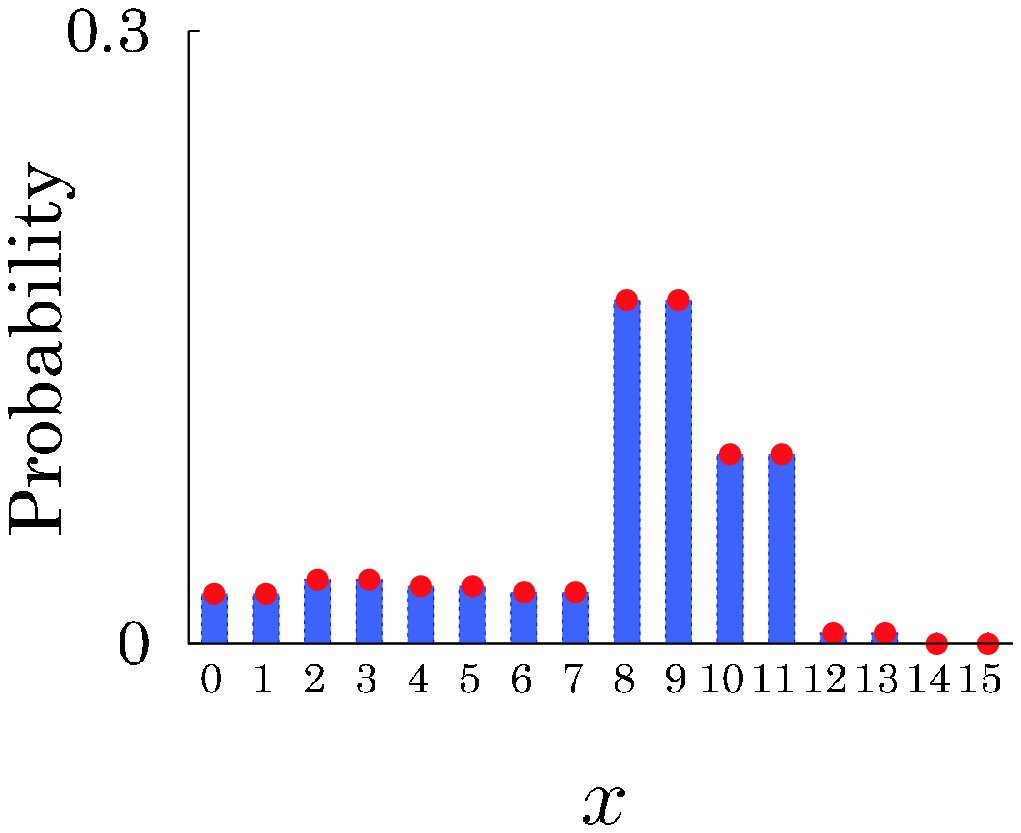}\\[2mm]
  (b) $\mathbb{P}(X_{15}^{HL}=x;1)$
  \end{center}
 \end{minipage}
 \begin{minipage}{35mm}
  \begin{center}
   \includegraphics[scale=0.3]{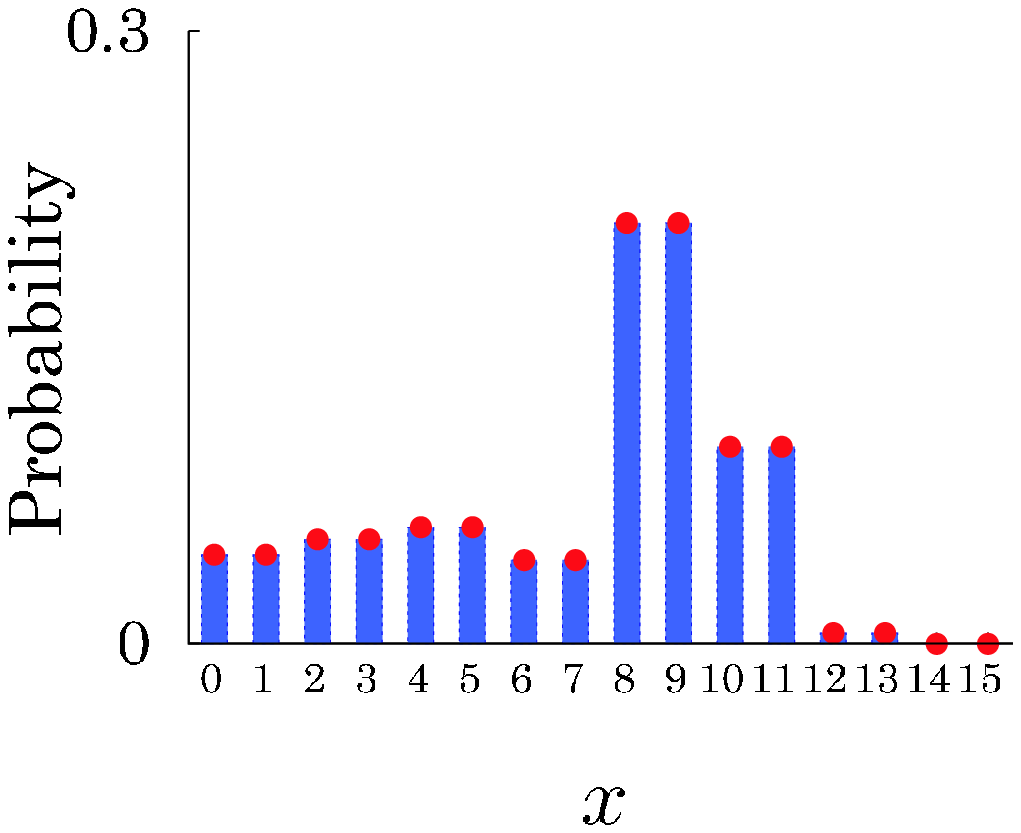}\\[2mm]
  (c) $\mathbb{P}(X_{15}^{HL}=x)$
  \end{center}
 \end{minipage}
\vspace{5mm}
\caption{$\theta=\pi/4$ : The results of the computation according to Eq. \eqref{eq:L_time-evolution} are represented by the blue bars. The red filled circles are estimated by Lemma \ref{lem:151209_7p} and Theorem \ref{th:151209_6} as $t=15$ which is an odd number.}
\label{fig:160201_16}
\end{center}
\end{figure}

\begin{figure}[h]
\begin{center}
 \begin{minipage}{35mm}
  \begin{center}
   \includegraphics[scale=0.3]{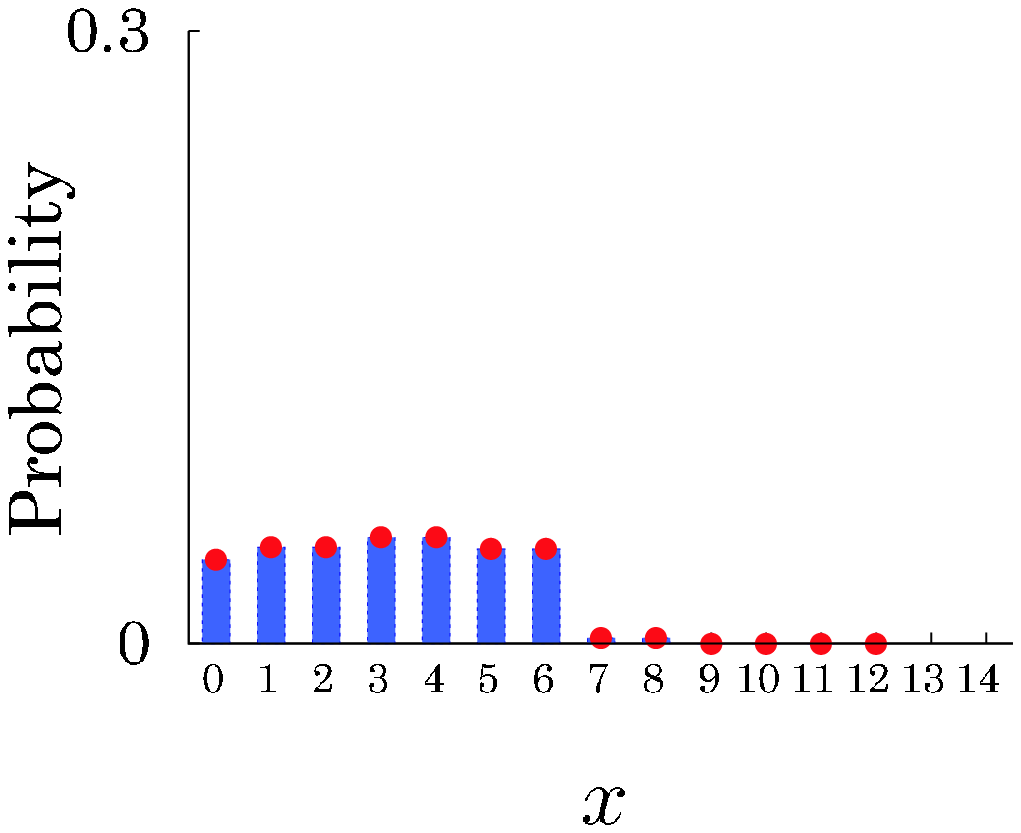}\\[2mm]
  (a) $\mathbb{P}(X_{14}^{HL}=x;0)$
  \end{center}
 \end{minipage}
 \begin{minipage}{35mm}
  \begin{center}
   \includegraphics[scale=0.3]{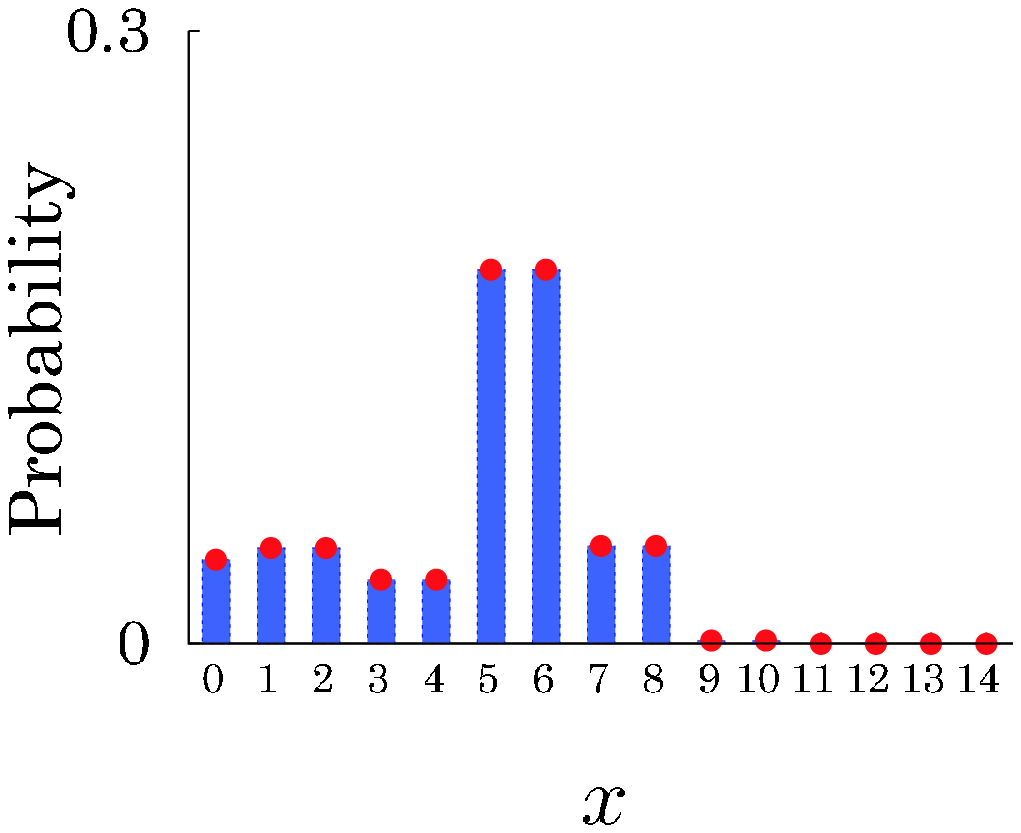}\\[2mm]
  (b) $\mathbb{P}(X_{14}^{HL}=x;1)$
  \end{center}
 \end{minipage}
 \begin{minipage}{35mm}
  \begin{center}
   \includegraphics[scale=0.3]{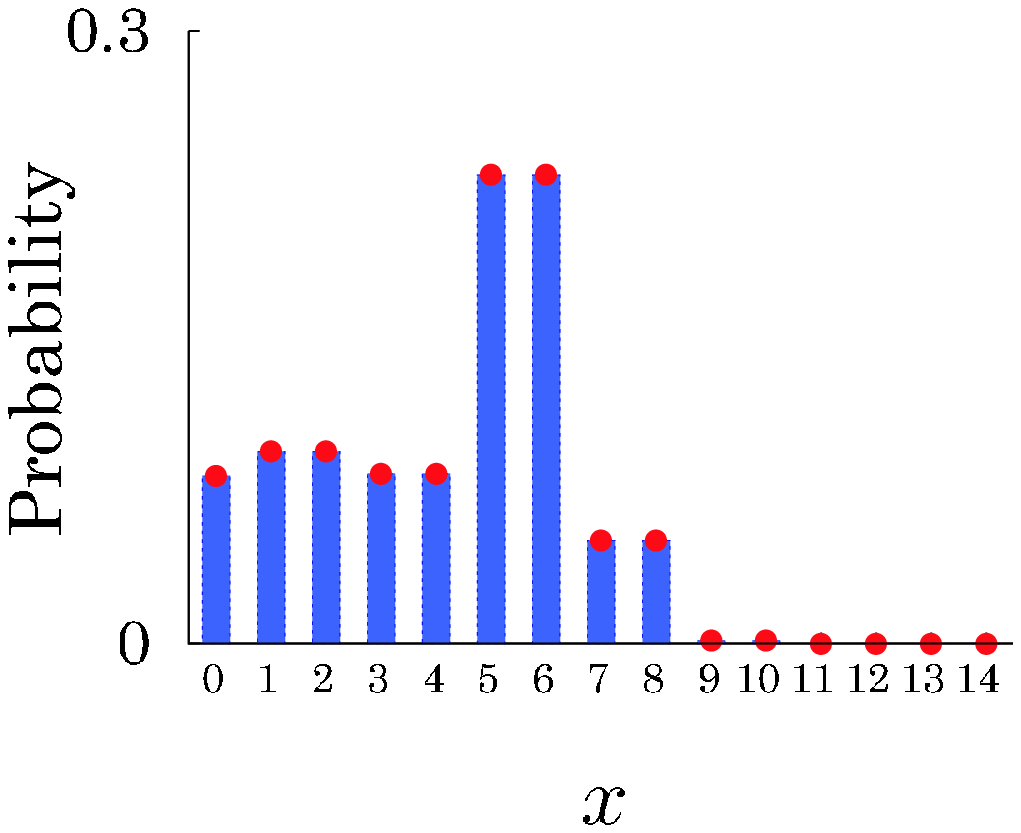}\\[2mm]
  (c) $\mathbb{P}(X_{14}^{HL}=x)$
  \end{center}
 \end{minipage}
\vspace{5mm}
\caption{$\theta=\pi/3$ : The results of the computation according to Eq. \eqref{eq:L_time-evolution} are represented by the blue bars. The red filled circles are estimated by Lemma \ref{lem:151209_7p} and Theorem \ref{th:151209_6} as $t=14$ which is an even number.}
\label{fig:160201_19}
\end{center}
\end{figure}

\begin{figure}[h]
\begin{center}
 \begin{minipage}{35mm}
  \begin{center}
   \includegraphics[scale=0.3]{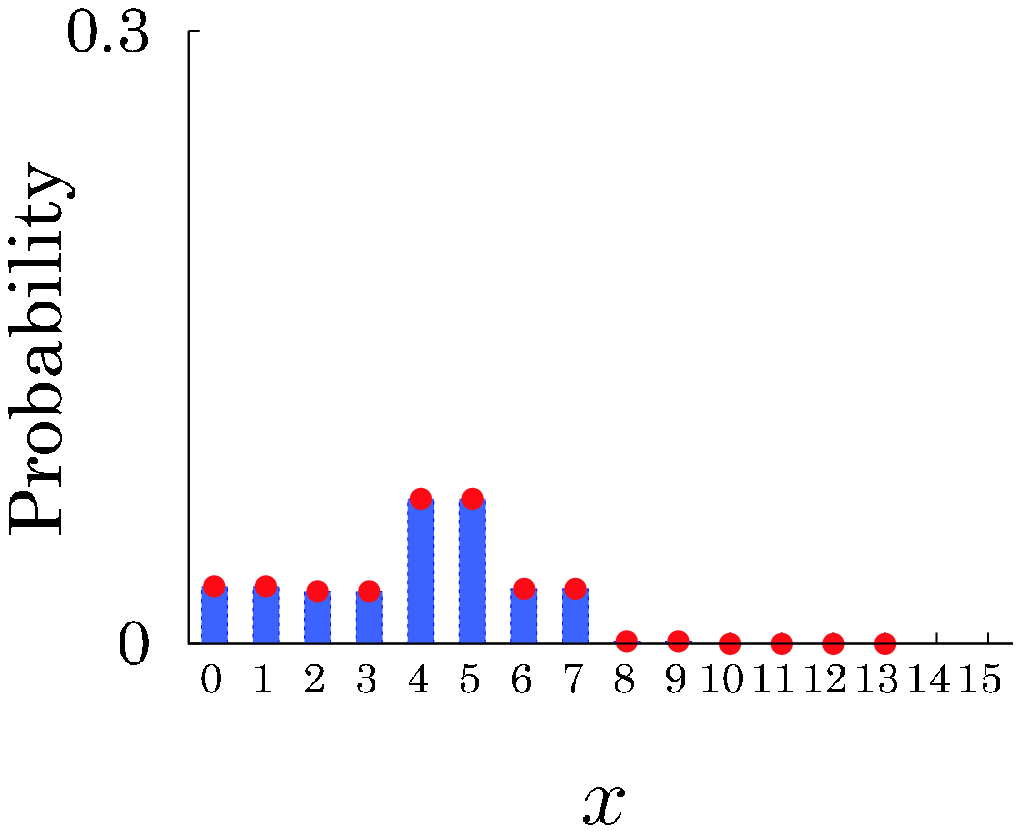}\\[2mm]
  (a) $\mathbb{P}(X_{15}^{HL}=x;0)$
  \end{center}
 \end{minipage}
 \begin{minipage}{35mm}
  \begin{center}
   \includegraphics[scale=0.3]{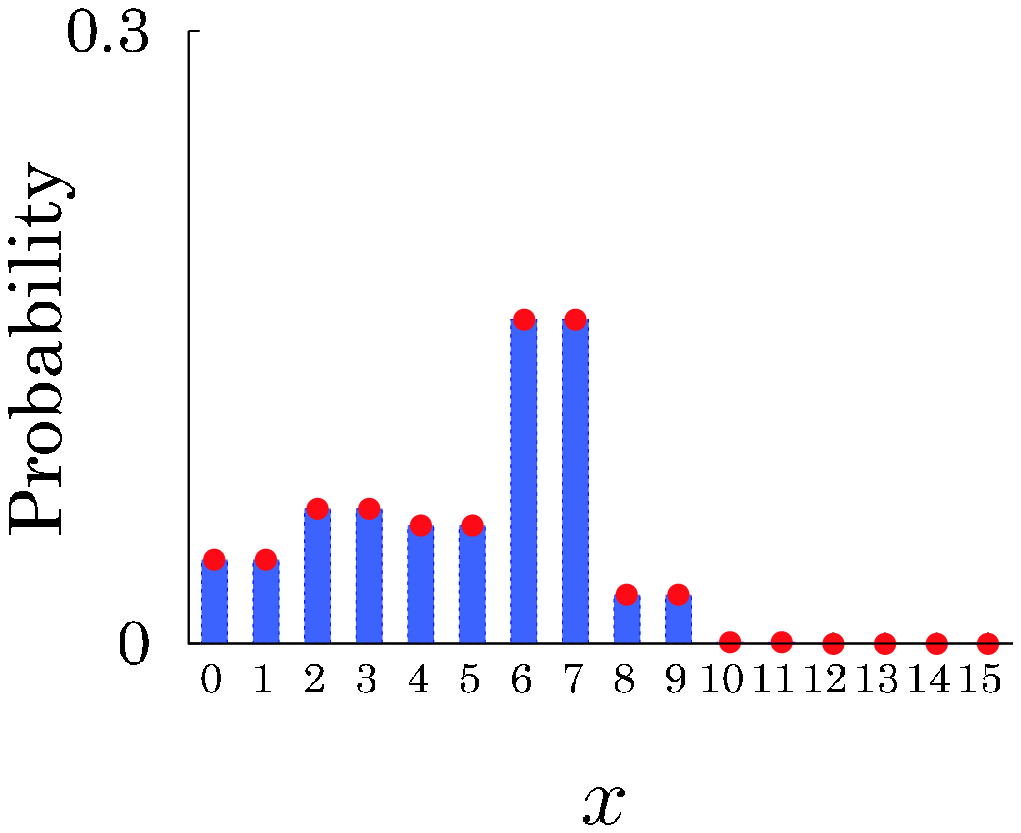}\\[2mm]
  (b) $\mathbb{P}(X_{15}^{HL}=x;1)$
  \end{center}
 \end{minipage}
 \begin{minipage}{35mm}
  \begin{center}
   \includegraphics[scale=0.3]{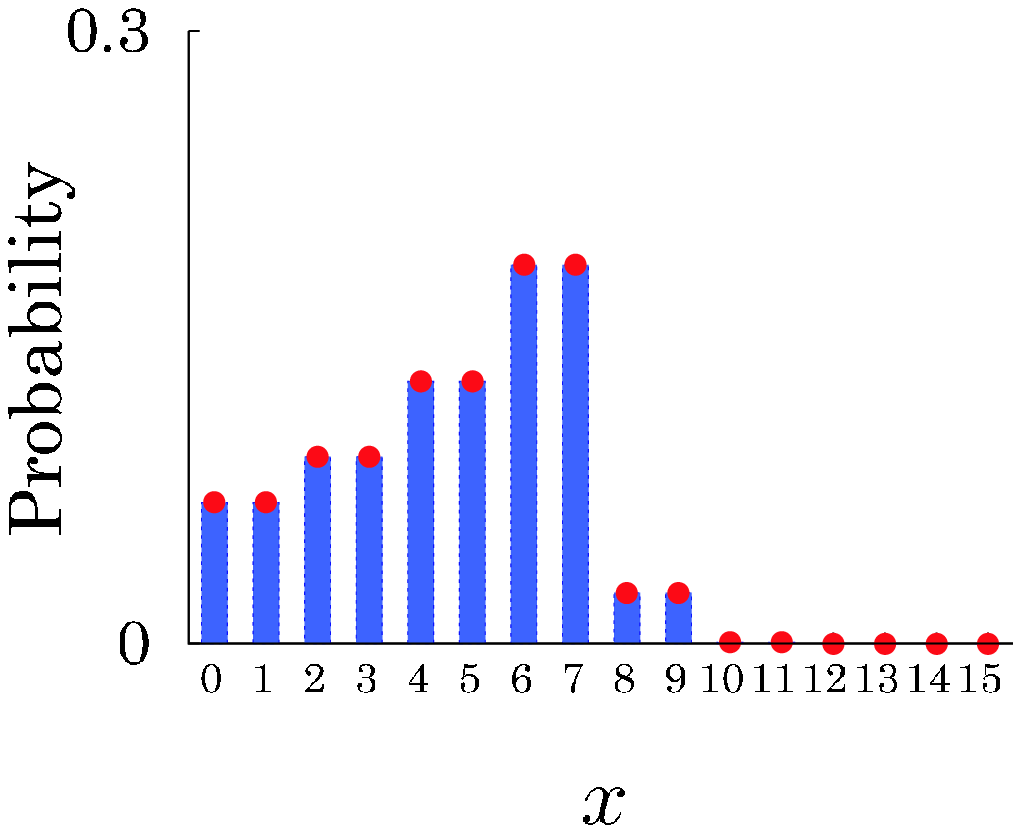}\\[2mm]
  (c) $\mathbb{P}(X_{15}^{HL}=x)$
  \end{center}
 \end{minipage}
\vspace{5mm}
\caption{$\theta=\pi/3$ : The results of the computation according to Eq. \eqref{eq:L_time-evolution} are represented by the blue bars. The red filled circles are estimated by Lemma \ref{lem:151209_7p} and Theorem \ref{th:151209_6} as $t=15$ which is an odd number.}
\label{fig:160201_22}
\end{center}
\end{figure}

Thanks to Lemma \ref{lem:151209_7p} and Theorem \ref{th:151209_6}, the Fourier analysis manages to work on the computation of a limit theorem for the quantum walk on the half line.
Since we know the limit distribution of the quantum walk on the line which was defined in Sec. \ref{sec:L_QW}, Eq. \eqref{eq:L_Fourier_analysis} tells us the following limit distributions as $t\to\infty$.

\begin{thm}
\label{th:151212}
 Assume that $\theta\neq 0,\pi/2,\pi,3\pi/2$. For a real number $x$, we have
 \begin{align}
  \lim_{t\to\infty}\mathbb{P}\left(\frac{X_t^{HL}}{t}\leq x;0\right)=&\int_{-\infty}^x \frac{|s|}{\pi(1+y)\sqrt{c^2-y^2}}I_{[0,|c|)}(y)\,dy,\label{eq:151212_7-1}\\
  \lim_{t\to\infty}\mathbb{P}\left(\frac{X_t^{HL}}{t}\leq x;1\right)=&\int_{-\infty}^x \frac{|s|}{\pi(1-y)\sqrt{c^2-y^2}}I_{[0,|c|)}(y)\,dy,\label{eq:151212_7-2}\\
  \lim_{t\to\infty}\mathbb{P}\left(\frac{X_t^{HL}}{t}\leq x\right)=&\int_{-\infty}^x \frac{2|s|}{\pi(1-y^2)\sqrt{c^2-y^2}}I_{[0,|c|)}(y)\,dy.\label{eq:151212_4}
 \end{align}
\end{thm}

\begin{proof}{%
 For a non-negative real number $x$, we derive Eqs. \eqref{eq:151212_7-1} and \eqref{eq:151212_7-2} from the result of the Fourier analysis in Eq. \eqref{eq:L_Fourier_analysis},
 \begin{align}
  &\lim_{t\to\infty}\mathbb{P}\left(\frac{X_t^{HL}}{t}\leq x;0\right)
  =\lim_{t\to\infty}\mathbb{P}\left(0\leq \frac{Y_t^L}{t}\leq x\right)\nonumber\\
  =&\int_0^x \frac{|s|}{\pi(1+y)\sqrt{c^2-y^2}}I_{(-|c|,|c|)}(y)\,dy
  =\int_0^x \frac{|s|}{\pi(1+y)\sqrt{c^2-y^2}}I_{[0,|c|)}(y)\,dy\nonumber\\
  =&\int_{-\infty}^x \frac{|s|}{\pi(1+y)\sqrt{c^2-y^2}}I_{[0,|c|)}(y)\,dy,
 \end{align}
 \begin{align}
  &\lim_{t\to\infty}\mathbb{P}\left(\frac{X_t^{HL}}{t}\leq x;1\right)
  =\lim_{t\to\infty}\mathbb{P}\left(-x\leq \frac{Y_t^L}{t}<0\right)\nonumber\\
  =&\int_{-x}^0 \frac{|s|}{\pi(1+y)\sqrt{c^2-y^2}}I_{(-|c|,|c|)}(y)\,dy\nonumber\\
  =&\int_0^x \frac{|s|}{\pi(1-y)\sqrt{c^2-y^2}}I_{(-|c|,|c|)}(-y)\,dy\nonumber\\
  =&\int_0^x \frac{|s|}{\pi(1-y)\sqrt{c^2-y^2}}I_{(-|c|,|c|)}(y)\,dy
  =\int_0^x \frac{|s|}{\pi(1-y)\sqrt{c^2-y^2}}I_{[0,|c|)}(y)\,dy\nonumber\\
  =&\int_{-\infty}^x \frac{|s|}{\pi(1-y)\sqrt{c^2-y^2}}I_{[0,|c|)}(y)\,dy.
 \end{align}
 On the other hand, for a negative real number $x$, these limits can be expressed as the same integral representations, 
 \begin{align}
  &\lim_{t\to\infty}\mathbb{P}\left(\frac{X_t^{HL}}{t}\leq x;0\right)=0
  =\int_{-\infty}^x \frac{|s|}{\pi(1+y)\sqrt{c^2-y^2}}I_{[0,|c|)}(y)\,dy,\\
  &\lim_{t\to\infty}\mathbb{P}\left(\frac{X_t^{HL}}{t}\leq x;1\right)=0
  =\int_{-\infty}^x \frac{|s|}{\pi(1-y)\sqrt{c^2-y^2}}I_{[0,|c|)}(y)\,dy.
 \end{align}
 We, therefore, figure out Eqs. \eqref{eq:151212_7-1} and \eqref{eq:151212_7-2} for any real number $x$.
 Adding up Eqs. \eqref{eq:151212_7-1} and \eqref{eq:151212_7-2}, one can obtain a limit theorem for the probability distribution $\mathbb{P}(X_t^{HL}=x)$, that is, Eq. \eqref{eq:151212_4}.
\qed
}
\end{proof}
\bigskip

When we assign the value $\pi/4$ to the parameter $\theta$, Eq. \eqref{eq:151212_4} outputs the limit distribution
\begin{equation}
 \lim_{t\to\infty}\mathbb{P}\left(\frac{X_t^{HL}}{t}\leq x\right)=\int_{-\infty}^x \frac{2}{\pi(1-y^2)\sqrt{1-2y^2}}I_{[0,1/\sqrt{2}\,)}(y)\,dy.
\end{equation}
This is consistent with a past study in Ref.~\cite{LiuPetulante2013} (See Corollary 2 in the paper).  
Theorem \ref{th:151212} allows us to hold the following approximations as time $t$ goes enough up,
\begin{align}
 \mathbb{P}(X_t^{HL}=x;0)\sim & \left\{\begin{array}{cl}
			       \frac{|s|t}{\pi (t+x)\sqrt{c^2t^2-x^2}} & (0 \leq x < |c|\,t ),\\[2mm]
				0 & (\mbox{otherwise}).
				     \end{array}\right.\label{eq:approximation_0}\\
 \mathbb{P}(X_t^{HL}=x;1)\sim & \left\{\begin{array}{cl}
			       \frac{|s|t}{\pi (t-x)\sqrt{c^2t^2-x^2}} & (0 \leq x < |c|\,t ),\\[2mm]
				0 & (\mbox{otherwise}).
				     \end{array}\right.\label{eq:approximation_1}\\
 \mathbb{P}(X_t^{HL}=x)\sim & \left\{\begin{array}{cl}
			       \frac{2|s|t^2}{\pi (t^2-x^2)\sqrt{c^2t^2-x^2}} & (0 \leq x < |c|\,t ),\\[2mm]
				0 & (\mbox{otherwise}).
				     \end{array}\right.\label{eq:approximation}
\end{align}
We should recall that the variable $x$ takes a non-negative integer, not a real number, in Eqs. \eqref{eq:approximation_0}, \eqref{eq:approximation_1}, and \eqref{eq:approximation}, which means $x\in\left\{0,1,2,\ldots\right\}$.
Figures \ref{fig:160201_25} and \ref{fig:160201_28} show the difference between the probability distribution and its approximation.
The blue lines form the probability distribution at time 500 by numerics and the red points indicate asymptotic values obtained by the usage of the right-sides in Eqs. \eqref{eq:approximation_0}, \eqref{eq:approximation_1}, and \eqref{eq:approximation} as $t=500$.

\begin{figure}[h]
\begin{center}
 \begin{minipage}{35mm}
  \begin{center}
   \includegraphics[scale=0.3]{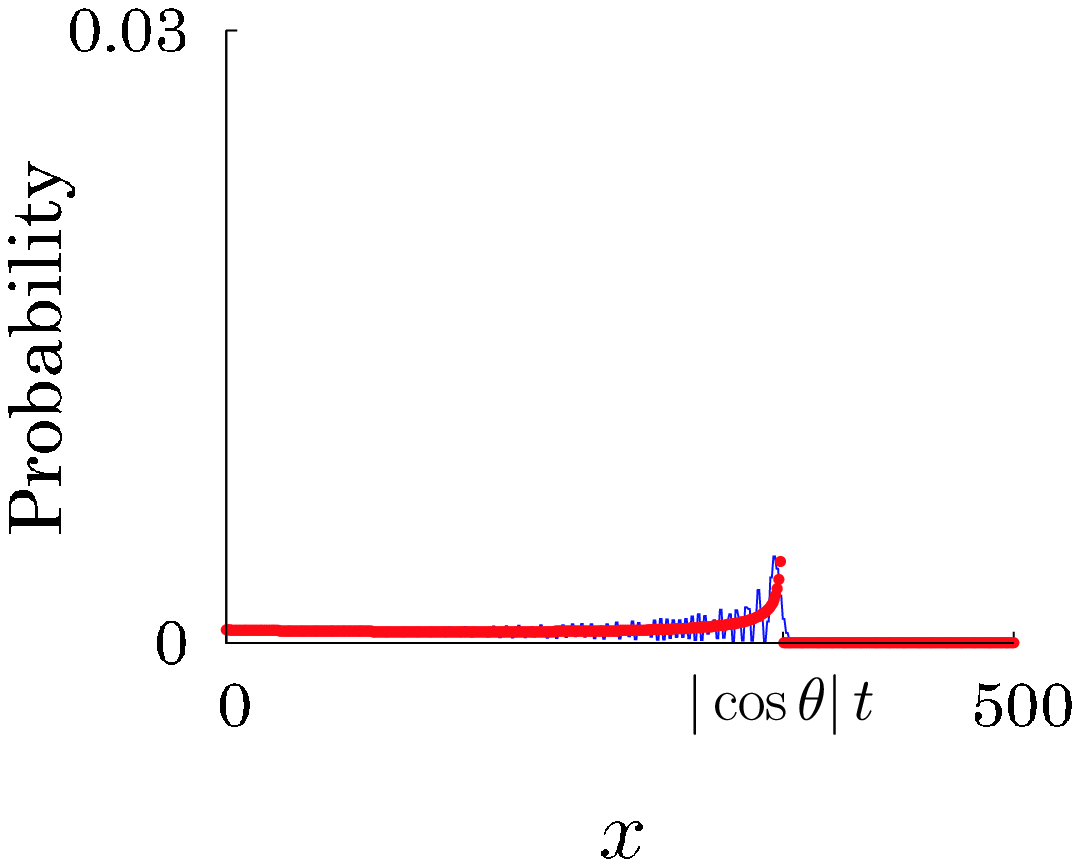}\\[2mm]
  (a) $\mathbb{P}(X_{500}^{HL}=x;0)$
  \end{center}
 \end{minipage}
 \begin{minipage}{35mm}
  \begin{center}
   \includegraphics[scale=0.3]{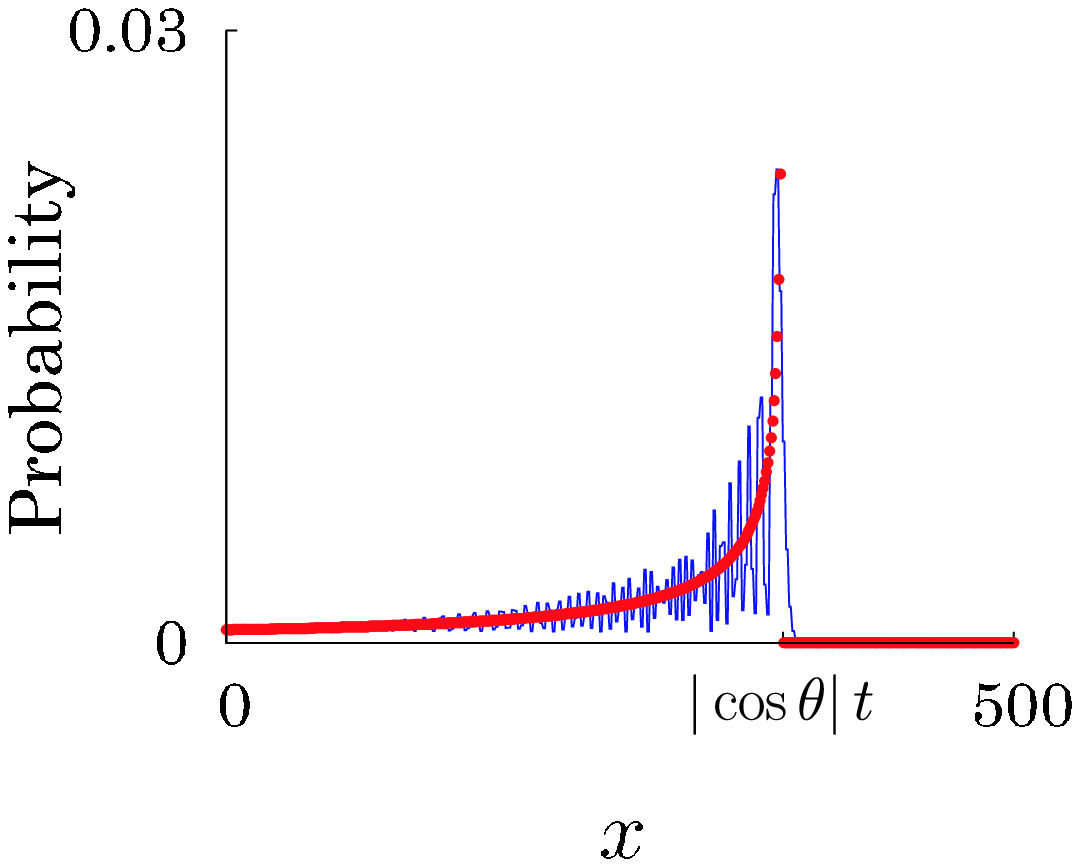}\\[2mm]
  (b) $\mathbb{P}(X_{500}^{HL}=x;1)$
  \end{center}
 \end{minipage}
 \begin{minipage}{35mm}
  \begin{center}
   \includegraphics[scale=0.3]{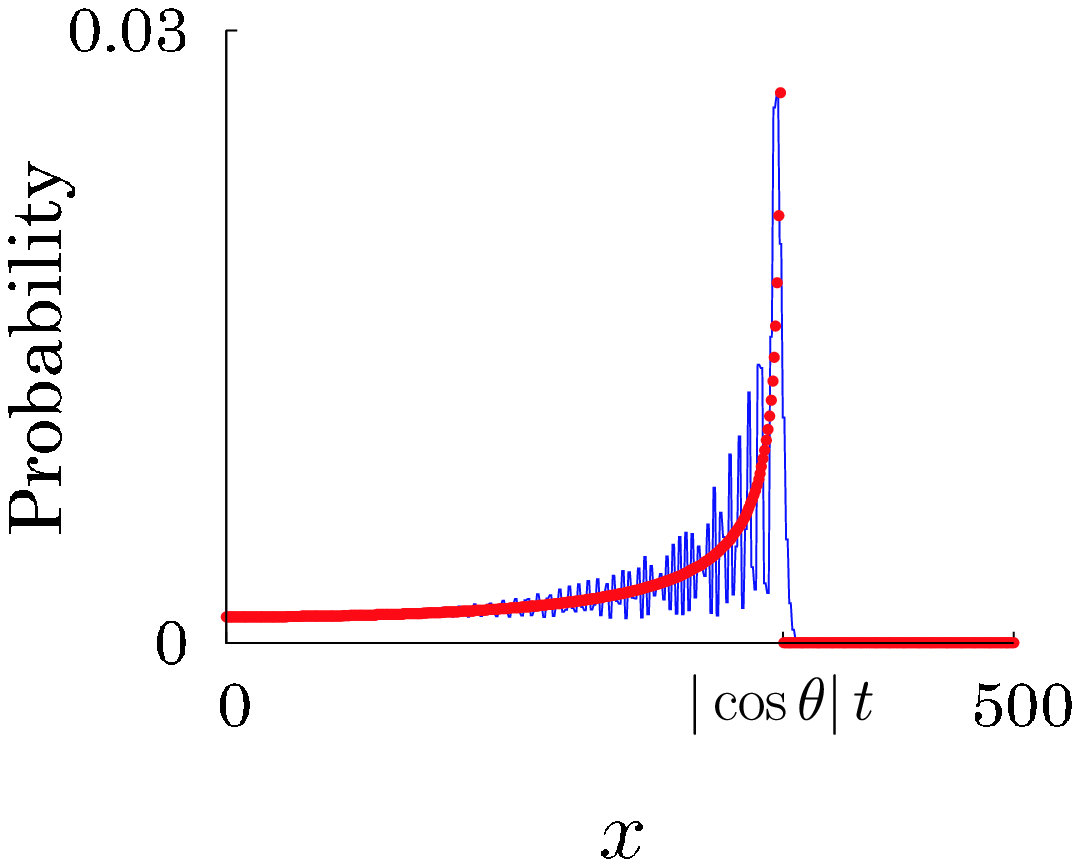}\\[2mm]
  (c) $\mathbb{P}(X_{500}^{HL}=x)$
  \end{center}
 \end{minipage}
\vspace{5mm}
\caption{$\theta=\pi/4$ : The blue lines represent the probability distribution at time 500. The red points indicate values obtained by the approximations in Eqs. \eqref{eq:approximation_0}, \eqref{eq:approximation_1}, and \eqref{eq:approximation} as $t=500$.}
\label{fig:160201_25}
\end{center}
\end{figure}

\begin{figure}[h]
\begin{center}
 \begin{minipage}{35mm}
  \begin{center}
   \includegraphics[scale=0.3]{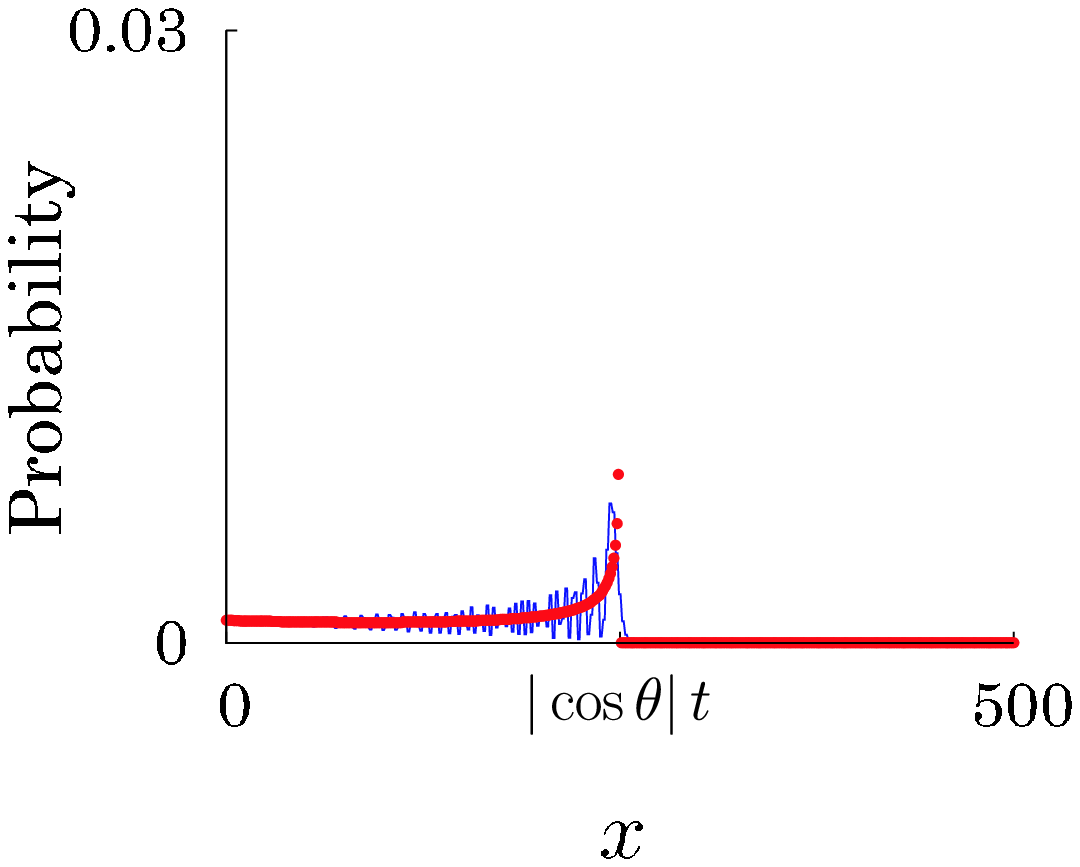}\\[2mm]
  (a) $\mathbb{P}(X_{500}^{HL}=x;0)$
  \end{center}
 \end{minipage}
 \begin{minipage}{35mm}
  \begin{center}
   \includegraphics[scale=0.3]{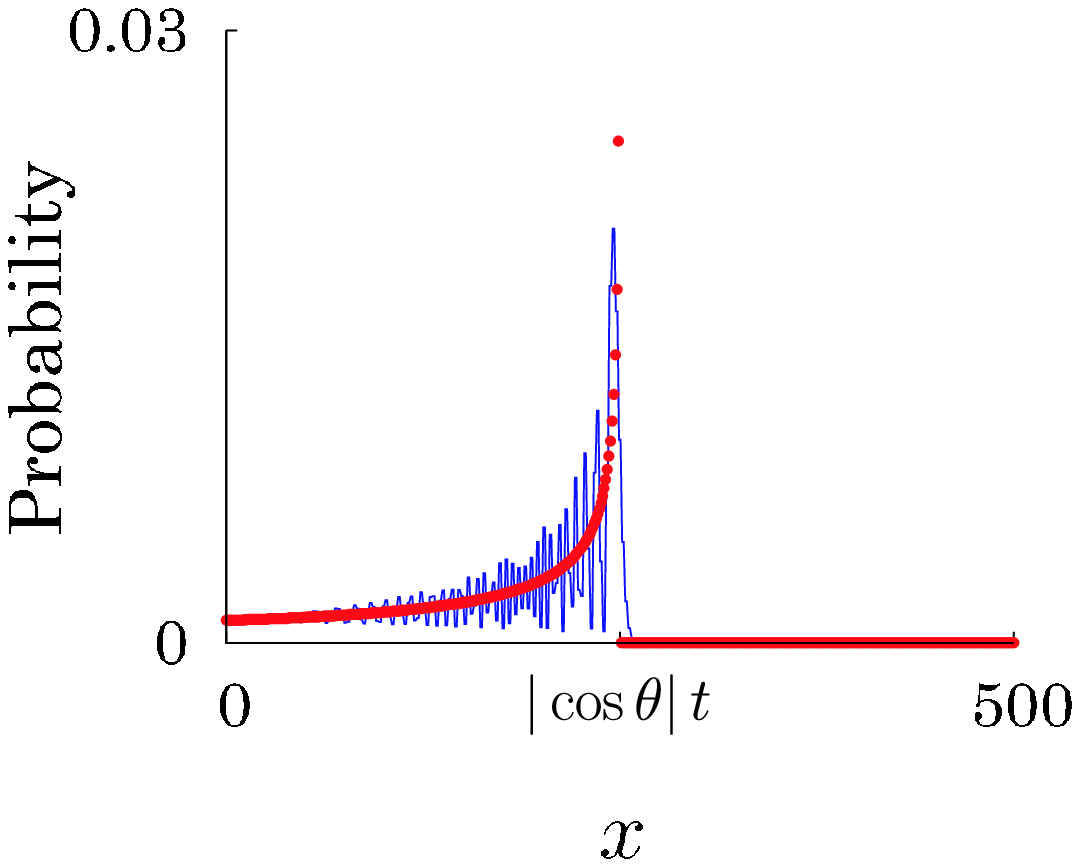}\\[2mm]
  (b) $\mathbb{P}(X_{500}^{HL}=x;1)$
  \end{center}
 \end{minipage}
 \begin{minipage}{35mm}
  \begin{center}
   \includegraphics[scale=0.3]{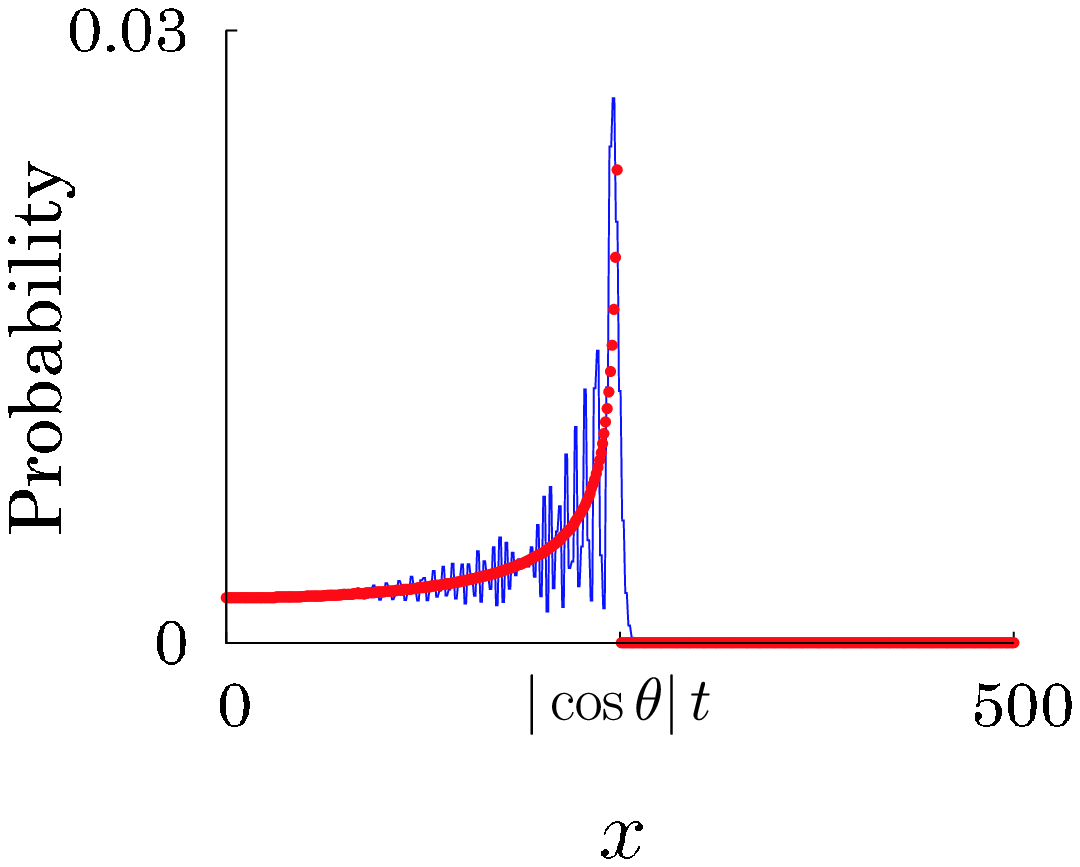}\\[2mm]
  (c) $\mathbb{P}(X_{500}^{HL}=x)$
  \end{center}
 \end{minipage}
\vspace{5mm}
\caption{$\theta=\pi/3$ : The blue lines represent the probability distribution at time 500. The red points indicate values obtained by the approximations in Eqs. \eqref{eq:approximation_0}, \eqref{eq:approximation_1}, and \eqref{eq:approximation} as $t=500$.}
\label{fig:160201_28}
\end{center}
\end{figure}

\section{Summary}
Motivated by the derivation of limit theorems for quantum walks on the half line by Fourier analysis, we took care of two quantum walks, a quantum walk on the half line and another one on the line.
Given a particular condition to the initial state of each quantum walk, the system of quantum walk on the line preserved all the information of the quantum walk on the half line, and vice versa.
As the result, it was figured out that the Fourier analysis became accessible to a limit theorem for the quantum walk on the half line by making the most of the copy onto the quantum walk on the line.
We can tell from Figs. \ref{fig:160201_25} and \ref{fig:160201_28} that the limit density function reproduces the probability distribution in approximation.
Also, we applied the result of a past study to the walk on the half line and got the exact representations of the probabilities $\mathbb{P}(X_t^{HL}=x;0)$, $\mathbb{P}(X_t^{HL}=x;1)$, and $\mathbb{P}(X_t^{HL}=x)$ in Lemma \ref{lem:151209_7p} and Theorem \ref{th:151209_6}.
They all completely agreed with the numerical experiments based on the time evolution rule of the quantum walk on the half line, as shown in Figs. \ref{fig:160201_13}--\ref{fig:160201_22}.
While the probabilities are exactly estimated by the representations, the approximations in Eqs. \eqref{eq:approximation_0}, \eqref{eq:approximation_1}, and \eqref{eq:approximation} are much easier to compute than the exact representations.
It, therefore, can be said that the limit distribution plays a role to inform roughly about the behavior of the quantum walk as time $t$ goes enough up.  

In this paper we had to go through a stage in which the quantum walk on the half line was copied onto the quantum walk on the line so that we proved Theorem \ref{th:151212} by Fourier analysis.
We, however, realized for sure a possibility that the Fourier analysis was available for the computation of limit theorems for quantum walks on the half line. 
Although the initial state was limited, the expansion should be done in future works.
As of now, it has been reported that some quantum walks on the half line can give rise to localization.
They were different from the quantum walk defined in this paper, but their limit distributions were derived by methods different from the Fourier analysis.
It is a challenging task to apply the Fourier analysis for such quantum walks.

\begin{center}
{\bf Acknowledgements}
\end{center}
The author is supported by JSPS Grant-in-Aid for Young Scientists (B) (No. 16K17648).


\end{document}